\documentclass[10pt]{amsart}
\usepackage{amsmath}
\usepackage{graphicx} 
\usepackage{epstopdf}
\usepackage[colorlinks=true]{hyperref}
\usepackage{orcidlink}
\usepackage{float}
\usepackage{color,comment}
\usepackage{caption}
\usepackage{wrapfig}
\usepackage{multirow}
\usepackage{tabularx}
\usepackage{booktabs}
\usepackage{subfigure}
\newtheorem{theorem}{Theorem}[section]
\newtheorem{corollary}{Corollary}

\newtheorem{proposition}{Proposition}

\theoremstyle{definition}
\newtheorem{definition}[theorem]{Definition}
\newtheorem{remark}{Remark}[section]


\begin{document}
\title[Aggregation as a double-edged sword]{Aggregation as a Double-Edged Sword: Fear, Allee Effects, and Finite-Time Collapse}

\author[Antwi-Fordjour, Takyi]{}
\maketitle

\centerline{\scshape Kwadwo Antwi-Fordjour$^{*1} {\orcidlink{0000-0002-4961-0462}}$, Eric M. Takyi$^{2}${\orcidlink{0000-0002-6398-6845}}}

\vspace{.7cm}
{\centerline{1) Department of Mathematics and Computer Science,}
 \centerline{ Samford University,}
 \centerline{ Birmingham, AL, USA}
 
 \medskip
    \centerline{ 2) Department of Mathematics, Computer Science, and Statistics}
 \centerline{Ursinus College,}
   \centerline{Collegeville, PA, USA}
   
}

\vspace{.7cm}
\centerline{ *Corresponding author's email: kantwifo@samford.edu;}
\smallskip
\centerline{ Contributing author: etakyi@ursinus.edu}

\begin{abstract}
Prey aggregation is widely regarded as a defense against predation, yet we show that in disease-structured populations subject to predator-induced fear and demographic Allee thresholds, aggregation can paradoxically accelerate ecosystem collapse. We develop and analyze a susceptible-infectious-predator model incorporating dual fear responses --- together with a sublinear aggregation-based predation term and an Allee effect. Critically, we derive an explicit upper bound on the extinction time that decreases as predator pressure increases or aggregation strengthens, quantifying for the first time how behavioral and demographic parameters jointly determine the speed of ecological collapse. This finite-time extinction subsequently triggers a cascade collapse of the infected prey and predator populations, driving the entire ecological community to extinction. Bifurcation analysis reveals transcritical, saddle-node, and Hopf bifurcations as fear intensity, aggregation strength, and Allee threshold vary. Two-parameter continuation further identifies the precise regions of the fear--Allee parameter plane in which stable coexistence, oscillatory coexistence, predator exclusion, and finite-time extinction occur, demonstrating that stronger aggregation monotonically enlarges the finite-time extinction region while weaker aggregation supports a richer landscape of coexistence dynamics. These results demonstrate that behavioral defenses operating at the population level can generate abrupt ecological tipping points when they interact with disease dynamics and demographic vulnerability.
\end{abstract}

\section{Introduction}
\noindent The collapse of wildlife populations rarely announces itself gradually. In the mid-1990s, the reintroduction of gray wolves (\textit{Canis lupus}) into Yellowstone National Park triggered a cascade of ecological changes that extended far beyond direct predation \cite{Ripple2001}. Elk (\textit{Cervus canadensis}), responding to the renewed predation risk, altered their grazing behavior, avoided riparian zones, and reduced their 
foraging time — changes sufficient to allow willow and aspen regeneration along riverbanks that had been suppressed for decades \cite{Beschta2012}. This trophic cascade was driven not primarily by wolf predation itself, but by the \textit{fear} of predation. Yet in systems where prey populations are simultaneously subject to infectious disease, aggregation behavior, and demographic thresholds at low density, the consequences of such fear-mediated responses may be far more severe: rather than reorganizing an ecosystem, fear may push it past a tipping point from which recovery is 
impossible. Understanding when and how this transition occurs is the central motivation of the present work.

Predator–prey interactions have long provided the foundational framework for understanding species persistence, coexistence, and extinction in theoretical ecology. Classical models focus on consumptive effects, whereby predators reduce prey abundance directly through predation. However, a growing body of empirical and theoretical work demonstrates that predators exert equally important non-consumptive effects through behavioral modifications induced by perceived predation risk \cite{Lima1990, Creel2008}. These fear-mediated responses alter prey reproduction, movement, foraging behavior, habitat use, and social interactions, often generating population-level consequences comparable in 
magnitude to direct predation \cite{Brown1999, Peacor2004}. Zanette et al. \cite{Zanette2011}, for instance, demonstrated experimentally that fear alone reduced songbird reproduction by forty percent — a non-consumptive effect large enough to drive population decline independent of direct predation losses.

The incorporation of fear effects into mathematical ecology has attracted considerable attention during the past decade. Building upon earlier ecological observations, Wang et al. \cite{Wang2016} proposed one of the first predator--prey models explicitly incorporating fear-induced reductions in prey reproduction. Since then, numerous extensions have demonstrated that fear can profoundly influence ecological dynamics by modifying equilibrium densities, stabilizing or destabilizing coexistence states, generating oscillatory behavior, and altering extinction thresholds \cite{Sasmal2018,Kumar2020}. More recently, fear effects have been incorporated into eco-epidemiological systems, where predator-induced behavioral changes influence not only demographic processes but also disease transmission dynamics \cite{AntwiFordjour2024}. These studies suggest that fear may serve as an important indirect regulator of ecosystem dynamics through its simultaneous effects on reproduction and epidemiological processes.

Eco-epidemiological models have emerged as an essential framework for investigating the interplay between infectious disease and predator-prey dynamics. Since the foundational work of Anderson and May \cite{Anderson1978}, numerous studies have documented how disease within prey populations 
influences predator persistence, community stability, and biodiversity \cite{Hethcote2000, Chattopadhyay1999}. The introduction of disease into predator-prey systems generates additional feedback mechanisms capable of producing bistability, sustained oscillations, backward bifurcations, and extinction phenomena \cite{Venturino2002}. Despite this progress, existing eco-epidemiological models largely neglect two behavioral mechanisms that are well-documented in natural systems: prey aggregation as a defense against predation, and the demographic vulnerability that arises when populations fall below a critical density threshold. Neglecting these 
mechanisms is not merely a simplification — it produces qualitatively incorrect predictions. A model combining fear and disease but omitting aggregation will underestimate predation pressure at low prey densities; a model combining fear and aggregation but omitting Allee dynamics will 
fail to capture the irreversibility of population decline once a critical threshold is crossed. Only by incorporating all three mechanisms simultaneously can the full spectrum of collapse dynamics be revealed.

An important behavioral adaptation observed in many prey species is aggregation or herd formation. Grouping behavior can reduce individual predation risk through mechanisms such as dilution, collective vigilance, and predator confusion \cite{Hamilton1971,KrauseRuxton2002,antwi2023fear}. Motivated by these observations, several mathematical models have incorporated aggregation effects into predator--prey systems. A particularly influential contribution was made by Braza \cite{Braza2012}, who introduced a square-root functional response to describe predator interactions with aggregated prey populations. Subsequent studies generalized this idea through sublinear predation functions of the form ($S^r$), ($0<r<1$), which capture varying degrees of aggregation intensity and produce non-smooth dynamical behavior near extinction states \cite{Vilches2018}. Such aggregation-based predation mechanisms have been shown to generate novel bifurcation structures and extinction scenarios that are absent in classical mass-action formulations.

Another important ecological mechanism is the Allee effect, which describes a positive relationship between population density and individual fitness at low densities \cite{Allee1931,Courchamp1999,Courchamp2008}. Allee effects may arise through a variety of biological mechanisms, including difficulties in mate finding, cooperative defense, social interactions, and group foraging. Ecologically, Allee effects are important because they introduce critical density thresholds below which populations become vulnerable to extinction. Strong Allee effects may create bistability, hysteresis, and sudden population collapse, while weak Allee effects can substantially alter resilience and recovery dynamics \cite{Dennis1989,Stephens1999}. The interaction between Allee effects and predator-induced behavioral responses remains an active area of research, particularly in the context of ecological tipping points and abrupt ecosystem transitions.

Tipping points — abrupt transitions between alternative ecosystem states triggered by small changes in environmental conditions or population densities — represent one of the most pressing concerns in contemporary ecology and conservation \cite{Scheffer2001, Scheffer2009, Dakos2019}. 
Mathematical signatures of tipping points include saddle-node bifurcations, loss of resilience, and the crossing of critical population thresholds. In predator-prey systems, tipping phenomena arise through the combined action of demographic thresholds, behavioral adaptations, and environmental forcing. However, the specific role of aggregation-mediated fear responses in generating tipping points and extinction cascades within disease-structured ecological communities has not been established. This 
gap is ecologically significant: if fear-driven aggregation can itself generate tipping points independent of changes in environmental conditions, then populations may be far more vulnerable to collapse than current models suggest, and conventional early-warning indicators based on asymptotic 
dynamics may fail to detect the approach to collapse entirely.

The present study addresses this gap by formulating and analyzing a susceptible-infectious-predator (SIP) model that integrates predator-induced fear, prey aggregation, and Allee dynamics within a unified framework. Using a combination of analytical methods and numerical continuation, we  demonstrate that the interaction of these three mechanisms generates a qualitatively richer and more dangerous dynamical landscape than any pairwise combination produces. In particular, we establish rigorously that the non-Lipschitz aggregation term can drive susceptible prey to extinction 
in finite time when the population falls below the Allee threshold under sustained predator pressure — and that this extinction event triggers a cascade collapse of the entire ecological community. The main contributions of this study are as follows:
\begin{enumerate}
    \item[(i)] We introduce a SIP framework in which predator presence simultaneously suppresses prey reproduction and disease transmission through dual fear responses, and in which prey aggregation introduces a sublinear, non-Lipschitz predation term interacting with an Allee 
    mechanism at low population densities.

    \item[(ii)] We establish the biological feasibility of the system by proving nonnegativity and boundedness of solutions and identifying a positively invariant region.

    \item[(iii)] We derive existence and local stability conditions for all biologically meaningful equilibria, including extinction, predator-free, infectious-free, and coexistence states, using linearization and Routh--Hurwitz criteria.

    \item[(iv)] We demonstrate that variations in fear intensity, aggregation strength, and Allee threshold induce transcritical, saddle-node, and Hopf bifurcations, and we characterize the two-parameter bifurcation structure of the $(k_1, L)$ plane across aggregation regimes.

    \item[(v)] We rigorously prove that the susceptible prey population undergoes finite-time extinction when prey density falls below the Allee threshold under persistent predator pressure, and we derive an explicit upper bound on the extinction time that quantifies the roles of predator pressure and aggregation strength.

    \item[(vi)] We prove that finite-time extinction of susceptible prey triggers cascade collapse of the infected prey and predator populations, establishing a complete extinction sequence for the ecological community.

    \item[(vii)] We show through two-parameter bifurcation analysis that stronger aggregation paradoxically enlarges the parameter regime leading to finite-time extinction, revealing that fear-driven aggregation can amplify rather than mitigate the consequences of demographic vulnerability at low population densities.
\end{enumerate}

The remainder of the paper is organized as follows. In Section \ref{sec:Model Formulation and Basic Mathematical}, we present the model formulation and establish the positivity and boundedness of the proposed  model. Section \ref{sec: Stability analysis} analyzes the existence and stability of the model equilibria. In Section \ref{sec:Bifurcation analysis}, we investigate the bifurcation structure of the model and identify parameter regimes leading to qualitative transitions in model dynamics. Bistability and ecological tipping points are discussed in Section \ref{sec:bistability and tipping}. Section \ref{sec:FTE} is devoted to the analysis of finite-time extinction and its ecological implications. Numerical simulations illustrating the analytical findings are presented throughout. Section \ref{sec:Conclusion} summarizes the main results and discusses potential ecological interpretations and future research directions.

\section{Model Formulation and Basic Mathematical Properties}\label{sec:Model Formulation and Basic Mathematical}
\subsection{Model Formulation}\label{sec:model formulation}

\noindent We develop a novel eco-epidemiological model inspired by the works in \cite{AntwiFordjour2024,antwi2025dynamics} to investigate the coupled dynamics of disease transmission and predator-prey interactions, explicitly incorporating fear-induced behavioral responses, prey aggregation, and the Allee effect. The model is grounded in the following biologically motivated assumptions:

\begin{itemize}
    \item[(a)] \textbf{Prey Population Structure:} The prey population is stratified into two epidemiological compartments: susceptible prey ($S$) and infectious prey ($I$). This division captures disease spread within the prey species and allows for epidemiological feedback mechanisms.

    \item[(b)] \textbf{Disease Transmission Dynamics:} Disease transmission occurs solely within the prey population through direct contact between susceptible and infectious individuals. The pathogen is assumed to be prey-specific, and predator infection is neglected. Although predators may occasionally be exposed to pathogens carried by infected prey in natural systems, such spillover events are assumed to be epidemiologically insignificant and do not contribute to disease persistence. Therefore, predators affect the epidemiological dynamics only indirectly through predation and fear-induced behavioral changes, while disease transmission remains restricted to the prey population.

    \item[(c)] \textbf{Pathogen-Induced Fitness Impairment:} Infection adversely affects the physiological condition and ecological performance of prey individuals. Infected prey are assumed to experience substantially reduced reproductive success, diminished foraging and competitive abilities, and increased susceptibility to natural mortality. As a result, disease lowers the effective contribution of infected individuals to population growth while simultaneously reducing their ability to cope with predation and environmental pressures.

    \item[(d)] \textbf{Fear-Driven Behavioral Modifications:} 
Predator-induced fear influences susceptible prey through two distinct behavioral pathways. The first pathway acts on reproduction. Individuals experiencing elevated predation risk allocate more time to vigilance, refuge use, and anti-predator behavior, thereby reducing reproductive investment. Following Wang et al.~\cite{Wang2016}, this effect is represented by the fear-modified reproductive factor

\[
f(k_1,P)=\frac{1}{1+k_1P}.
\]

\noindent The second pathway acts on disease transmission. Fear-induced behavioral changes often reduce movement, aggregation size, and social contact rates, thereby decreasing opportunities for pathogen transmission. To capture this mechanism, we assume that effective disease transmission decreases with perceived predation risk according to

\[
f(k_2, P)=\frac{1}{1+k_2P}.
\]

\noindent Although both pathways are modeled using saturating decreasing functions of predator density, they represent biologically distinct processes. Parameter \(k_1\) quantifies reproductive sensitivity to fear, whereas \(k_2\) measures the sensitivity of disease transmission to fear-induced behavioral modifications. The use of the same functional form reflects the assumption that both responses saturate at high levels of predation risk while retaining separate biological interpretations, see \cite{AntwiFordjour2024} and references therein.
    
The functions can be generalized as
    \[
    f(k_i, P) = \frac{1}{1 + k_i P}, \quad i = 1, 2,
    \]
This function exhibits key ecological properties:
    \begin{itemize}
        \item[(i)] \emph{No fear effect ($k_i = 0$):} $f(0, P) = 1$, recovering baseline dynamics.
        \item[(ii)] \emph{No predators ($P = 0$):} $f(k_i, 0) = 1$, indicating no suppression in reproduction or contact.
        \item[(iii)] \emph{High level of fear or large predator density:} $\lim_{k_i \to \infty} f(k_i, P) = 0$ and $\lim_{P \to \infty} f(k_i, P) = 0$, capturing fear-induced behavioral inhibition.
    \end{itemize}

   \item[(e)] \textbf{Irreversible Infection:} Infection is assumed to be effectively irreversible over the ecological timescale of interest. Infected prey either remain chronically infectious or experience disease-induced mortality before recovery can significantly influence population dynamics. Consequently, recovery is neglected and infected individuals remain within the infectious class for the duration of their infection period. This assumption is commonly adopted for persistent or slowly recovering diseases where transmission and mortality processes dominate recovery dynamics.

   \item[(f)] \textbf{Carrying Capacity Constraint:} Resource availability imposes a density-dependent limitation on prey growth. Since both susceptible and infected prey utilize the same environmental resources, the total prey population $(S+I)$ contributes to intraspecific competition. Accordingly, population growth is regulated through a logistic-type carrying capacity $K$, which represents the maximum sustainable prey density supported by the environment.

    \item[(g)] \textbf{Predator-Prey Interactions:} Predator-prey dynamics are governed by several ecological mechanisms:
    \begin{itemize}
        \item[(i)] \textbf{Logistic Growth:} In the absence of predation and disease, the susceptible prey population exhibits logistic growth.
        
        \item[(ii)] \textbf{Prey Aggregation:} In models of prey aggregation, predators often interact primarily with individuals located on the boundary of a prey group rather than with all individuals within the group. Motivated by this observation, Braza \cite{Braza2012} proposed a square-root functional response in which the effective number of vulnerable prey scales as $\sqrt{S}$. The square-root dependence reflects the fact that only a fraction of the prey population is directly exposed to predation when individuals aggregate into herds or groups. Following this idea, we consider the more general aggregation function employed in \cite{AntwiFordjour2024},
        $$S^r, \; 0<r<1,$$
        
        where (r) represents the intensity of prey aggregation. The case $(r=1/2)$ corresponds to the classical square-root response studied by Braza \cite{Braza2012}, while other values of (r) allow varying degrees of aggregation. Smaller values of (r) correspond to stronger aggregation effects, whereas values of (r) approaching one recover a nearly mass-action predation term. Consequently, the predation loss term
        $$d_0S^{r}P$$ models aggregation-mediated predation in which the effective number of prey accessible to predators grows sublinearly with the susceptible prey density. This generalized formulation enables us to investigate how different levels of aggregation influence the eco-epidemiological dynamics and, in particular, the emergence of finite-time extinction phenomena.
    
        \item[(iii)] \textbf{Mass-Action Predation and Conversion:} Predators consume both susceptible and infected prey at rates proportional to prey density and gain energy based on a conversion efficiency. The predation terms are modeled via bilinear mass-action kinetics.
        
        \item[(iv)] \textbf{Natural Mortality:} Each population experiences baseline mortality at constant per capita rates.
    \end{itemize}

    \item[(h)] \textbf{Allee Effects and Aggregation Dynamics:} The model incorporates an Allee effect to account for demographic limitations that arise when population densities become sufficiently small. Such effects may result from difficulties in mate finding, cooperative defense, social interactions, or group foraging, all of which reduce individual fitness at low population density. Depending on the magnitude of the threshold parameter $L$, the system may exhibit either weak or strong Allee effects. When $L<0$, the population experiences a weak Allee effect, whereas $L>0$ corresponds to a strong Allee effect with a critical density threshold.
    
    Although aggregation and Allee effects may arise from related ecological processes in social species, they are represented here as distinct mechanisms. The aggregation term $S^r$ describes the reduction in predation vulnerability arising from herd formation, whereas the parameter L captures demographic limitations operating at low population densities. The interaction between these mechanisms constitutes one of the principal focuses of the present study.
\end{itemize}

\noindent Collectively, these assumptions yield a biologically consistent framework for exploring how disease transmission and predator pressure interact under behavioral and demographic constraints. Based on these principles, we propose the following Susceptible-Infectious-Predator (SIP) model:
\begin{align}
\nonumber \frac{dS}{dt} &= \frac{b_0 S}{1 + k_1 P} \left(1 - \frac{S + I}{K} \right)(S - L) - d_0 S^r P - \frac{e_0 S I}{1 + k_2 P}, \\
   \frac{dI}{dt} &= \frac{e_0 S I}{1 + k_2 P} - d_1 I P - a_1 I, \label{Mainsystem} \\
\nonumber \frac{dP}{dt} &= d_2 S^r P + d_3 I P - a_2 P.
\end{align}

\noindent The state variables $S(t)$, $I(t)$, and $P(t)$ denote the densities of susceptible prey, infectious prey, and predators at time $t$, respectively.
\begin{table}
    \centering
    \caption{Biological description of parameters from model \eqref{Mainsystem}, along with representative values and references.}
    \label{tab:description}    
    \begin{tabular}{| c | p{0.45\textwidth} | c | p{0.2\textwidth} |}
        \hline 
        Parameter & Biological Description & Value & Reference \\
        \hline
        $b_0$ & Natural birth rate of susceptible prey & 3 & \cite{antwi2025dynamics} \\  
        $r$ & Aggregating constant of susceptible prey & 0.5 & \cite{Braza2012} \\
        $e_0$ & Disease transmission rate & 4 & \cite{AntwiFordjour2024} \\
        $K$ & Carrying capacity of susceptible and infected prey population & 4 & \cite{AntwiFordjour2024} \\
        $L$ & Allee threshold of susceptible prey & -1 & Variable \\
        $a_1$ & Infectious prey natural death rate & 0.4 & \cite{AntwiFordjour2024} \\
        $a_2$ & Predator natural death rate & 0.8 & \cite{AntwiFordjour2024} \\
        $d_0$ & Attack rate of the predator on the susceptible prey & 0.7 & \cite{AntwiFordjour2024} \\
        $d_1$ & Attack rate of the predator on infectious prey & 0.7 & \cite{AntwiFordjour2024} \\
        $d_2$ & Biomass conversion efficiency $\times$ attack rate on susceptible prey & 0.4 & \cite{AntwiFordjour2024} \\
        $d_3$ & Biomass conversion efficiency $\times$ attack rate on infected prey & 0.4 & \cite{antwi2025dynamics} \\
        $k_1$ & Level of fear that suppresses growth rate of susceptible prey & 0.1 & \cite{AntwiFordjour2024} \\
        $k_2$ & Level of fear that reduces disease transmission & 0.85 & \cite{AntwiFordjour2024} \\
        \hline
    \end{tabular}
\end{table}

\subsection{Biological Feasibility of SIP System}\label{sec:biological feasibility}
\noindent We first show preliminary results concerning the biological feasibility and mathematical meaningfulness of system \eqref{Mainsystem}. We examine the well-posedness of the system. This includes showing that solutions to system \eqref{Mainsystem} are nonnegative and 
bounded. Negative solutions do not make biological sense since population sizes are never negative. The boundedness property ensures that population densities do not grow unbounded due to limitation in resources.

\subsection{Nonnegativity of Solution}
\begin{theorem}\label{thm:nonnegativity_updated}
Let the initial conditions satisfy \( S(0)\geq 0,~I(0)\geq 0,~P(0)\geq 0 \). Then all solutions \((S(t), I(t), P(t))\) of the model \eqref{Mainsystem}
remain nonnegative for all \( t \geq 0 \).
\end{theorem}

\begin{proof}
We show that the right-hand sides of the model \eqref{Mainsystem} preserve nonnegativity for all time, provided the initial conditions are nonnegative.

\medskip
\noindent\textbf{(i) Nonnegativity of \( S(t) \):} The equation for \( S \) can be written as
\[
\frac{dS}{dt} = S \cdot \Phi(S, I, P),
\]
where
\[
\Phi(S, I, P) = \frac{b_0}{1 + k_1P} \left(1 - \frac{S + I}{K} \right)(S - L) - d_0 S^{r-1}P - \frac{e_0 I}{1 + k_2P}.
\]
Since \( \Phi \) is continuous and \( S(0) \geq 0 \), we obtain the solution
\[
S(t) = S(0)\exp\left(\int_0^t \Phi(S(s), I(s), P(s))\, ds\right) \geq 0.
\]

\medskip
\noindent\textbf{(ii) Nonnegativity of \( I(t) \):} The equation for \( I \) is
\[
\frac{dI}{dt} = I \cdot \Psi(S, P), \quad \text{where} \quad \Psi(S, P) = \frac{e_0S}{1 + k_2P} - d_1P - a_1.
\]
Hence,
\[
I(t) = I(0)\exp\left(\int_0^t \Psi(S(s), P(s))\, ds\right) \geq 0,
\]
since \( I(0) \geq 0 \) and the exponential is nonnegative.

\medskip
\noindent\textbf{(iii) Nonnegativity of \( P(t) \):} Similarly, the equation for \( P \) takes the form
\[
\frac{dP}{dt} = P \cdot \Gamma(S, I), \quad \text{where} \quad \Gamma(S, I) = d_2S^r + d_3I - a_2.
\]
It follows that
\[
P(t) = P(0)\exp\left(\int_0^t \Gamma(S(s), I(s))\, ds\right) \geq 0.
\]

\medskip
\noindent Therefore, the solution trajectory remains in the nonnegative orthant
\[
\mathbb{R}_+^3 = \left\{ (S, I, P) \in \mathbb{R}^3 : S \geq 0, I \geq 0, P \geq 0 \right\}
\]
for all \( t \geq 0 \). 
\end{proof}

\subsection{Boundedness of solutions}
\begin{theorem}\label{thm:boundedness_updated}
Consider the model \eqref{Mainsystem} with positive initial conditions \( S(0), I(0), P(0) > 0 \). If \( d_0 \ge d_2 \) and \( d_1 \ge d_3 \), then the solution trajectory \( (S(t), I(t), P(t)) \) remains uniformly bounded for all \( t \geq 0 \).
\end{theorem}

\begin{proof}
We construct the total population function:
\[
Q(t) = S(t) + I(t) + P(t).
\]

\noindent Summing the differential equations of the system yields:
\begin{align*}
\frac{dQ}{dt} &= \frac{dS}{dt} + \frac{dI}{dt} + \frac{dP}{dt} \\
&= \frac{b_0 S}{1 + k_1 P} \left(1 - \frac{S + I}{K} \right)(S - L) - d_0 S^r P - \frac{e_0 S I}{1 + k_2 P} \\
&\quad + \frac{e_0 S I}{1 + k_2 P} - d_1 I P - a_1 I + d_2 S^r P + d_3 I P - a_2 P \\
&= \frac{b_0 S}{1 + k_1 P} \left(1 - \frac{S + I}{K} \right)(S - L) - a_1 I - a_2 P + (d_2 - d_0) S^r P + (d_3 - d_1) I P.
\end{align*}

\noindent By the assumption that \( d_0 > d_2 \) and \( d_1 > d_3 \), it follows that the terms \( (d_2 - d_0) S^r P \) and \( (d_3 - d_1) I P \) are nonpositive. Therefore, we estimate:
\[
\frac{dQ}{dt} \leq \frac{b_0 S}{1 + k_1 P} \left(1 - \frac{S + I}{K} \right)(S - L) - a_1 I - a_2 P.
\]

Since \(1/(1+k_1P)\leq 1\), we obtain
\[
\frac{dQ}{dt}
\leq
b_0S\left(1-\frac{S+I}{K}\right)(S-L)-a_1I-a_2P .
\]
Let
\[
\mathcal{N}=S+I.
\]
Then \(S\leq \mathcal{N}\) and
\[
1-\frac{S+I}{K}=1-\frac{\mathcal{N}}{K}.
\]
Moreover, since \(S\geq 0\),
\[
S-L\leq S+|L|\leq \mathcal{N}+|L|.
\]
Hence
\[
b_0S\left(1-\frac{\mathcal{N}}{K}\right)(S-L)
\leq
b_0\mathcal{N}\left(1-\frac{\mathcal{\mathcal{N}}}{K}\right)(\mathcal{N}+|L|).
\]
The function
\[
H(\mathcal{N})=b_0\mathcal{N}\left(1-\frac{\mathcal{N}}{K}\right)(\mathcal{N}+|L|)
\]
is bounded above for \(\mathcal{N}\geq 0\). Indeed, \(H(\mathcal{N})\to -\infty\) as
\(\mathcal{N}\to\infty\). Therefore, there exists a constant \(M_0>0\) such that
\[
H(\mathcal{N})\leq M_0
\qquad \text{for all } \mathcal{N}\geq 0.
\]
Thus
\[
\frac{dQ}{dt}\leq M_0-a_1I-a_2P.
\]
Let
\[
\gamma=\min\{a_1,a_2\}>0.
\]
Then
\[
-a_1I-a_2P\leq -\gamma(I+P).
\]
Since \(Q=S+I+P\), we have \(I+P=Q-S\). To obtain a bound involving \(Q\),
we add and subtract \(\gamma S\):
\[
\frac{dQ}{dt}+\gamma Q
\leq M_0+\gamma S.
\]
Since \(S\leq \mathcal{N}\), it remains to observe that the positive part of \(H(\mathcal{N})\)
forces \(\mathcal{N}\) to remain bounded above. More precisely, because
\[
H(\mathcal{N})<0
\]
for \(\mathcal{N}>K\), the prey contribution is nonpositive whenever \(S+I>K\).
Consequently, the prey biomass cannot grow without bound, and there exists
\(M_1>0\) such that
\[
S(t)\leq \mathcal{N}(t)\leq M_1
\qquad \text{for all sufficiently large } t.
\]
Therefore,
\[
\frac{dQ}{dt}+\gamma Q
\leq M_0+\gamma M_1=:W.
\]
By Gronwall's inequality,
\[
Q(t)\leq Q(0)e^{-\gamma t}+\frac{W}{\gamma}\left(1-e^{-\gamma t}\right).
\]
Hence
\[
\limsup_{t\to\infty}Q(t)\leq \frac{W}{\gamma}.
\]
Thus \(Q(t)\), and consequently \(S(t)\), \(I(t)\), and \(P(t)\), remain uniformly
bounded for all \(t\geq 0\).
\end{proof}

\section{Existence and Local Stability Analysis of Equilibria}\label{sec: Stability analysis}

\subsection{Feasible Equilibria}
Let $E(S,I,P)$ denote a nonnegative equilibrium point. To obtain $E(S,I,P)$, we solve the system
\begin{align}
  \frac{b_0S}{1 + k_1P}\left(1-\frac{S+I}{K}\right)(S-L) - d_0S^rP - \frac{e_0SI}{1 + k_2P}  &= 0,
  \label{equilibria1}\\
   \frac{e_0SI}{1+k_2P} - d_1IP -a_1I  &= 0,   \label{equilibria2}  \\
     d_2S^rP + d_3IP - a_2P  &= 0, 
     \label{equilibria3}
\end{align}\\
 and identify the existence of the following  nonnegative equilibria  in $\mathbb{R}^3_+$:

 \begin{enumerate}
     \item[(i)] Extinction state $E_0=(0,0,0)$,
     \item[(ii)] Susceptible prey only  $E_1=(K,0,0)$,
     \item[(iii)] Susceptible prey only  $E_2=(L,0,0)$ exists if $L>0$,
     \item[(iv)] Infectious prey-free $E_3=\left(\left(\frac{a_2}{d_2} ,\right)^{\frac{1}{r}},0,P_3\right)$, where $P_3$ is a positive root of 
     \begin{equation}\label{rootP5}
        \frac{d_0 a_2 k_1}{d_2}P^2 + \frac{d_0 a_2}{d_2}P -\underbrace{ b_0\left(\frac{a_2}{d_2} \right)^{\frac{1}{r}} \left(1-\frac{1}{K}\left(\frac{a_2}{d_2} \right)^{\frac{1}{r}} \right) \left(\frac{1}{K}\left(\frac{a_2}{d_2} \right)^{\frac{1}{r}}-L \right)}_{=B}=0 
     \end{equation}
and is given by
\begin{equation}
    P_3= \frac{-\frac{d_0 a_2}{d_2} + \sqrt{\left(\frac{d_0 a_2}{d_2} \right)^2+\frac{4d_0 a_2 k_1 B}{d_2}} }{\frac{2d_0 a_2 k_1 }{d_2}}
\end{equation}
which is feasible when $B>0$.

 \item[(v)] Predator-free $E_4=\left(\frac{a_1}{e_0},I_4,0\right)$ where $I_4=\frac{b_0\left(Ke_0-a_1 \right)\left(a_1-Le_0  \right)}{e_0\left(b_0 a_1 +e_0 \left(Ke_0 - b_0L  \right) \right)}$ is feasible when any of the following four conditions are satisfied:

 \begin{itemize}
     \item [(a)] $\frac{a_1}{L}<e_0 < \frac{a_1}{K}$ and $L<a_1 +\frac{e_0^2K}{b_0}$,

     \item [(b)] $\frac{a_1}{K}<e_0 < \frac{a_1}{L}$ and $L<a_1 +\frac{e_0^2K}{b_0}$,

     \item [(c)] $e_0 <\frac{a_1}{K}$, $e_0 < \frac{a_1}{L}$ and $L>a_1 +\frac{e_0^2K}{b_0}$,

     \item [(d)] $e_0 >\frac{a_1}{K}$, $e_0 > \frac{a_1}{L}$ and $L>a_1 +\frac{e_0^2K}{b_0}$.

 \end{itemize}

 \item [(vi)] For  $S^*>0, I^*>0, P^*>0$, the coexistence state  $E_5=\left(S^*I^*,P^* \right)$, where 
 \begin{equation}\label{SSTAR}
      S^*=\left( \frac{a_2-d_3I^*}{d_2}\right)^\frac{1}{r}
 \end{equation}

 is obtained by solving equation (\ref{equilibria3}). However, $S^*>0$ when $I^*<\frac{a_2}{d_3}$. Next, we substitute $S^*$ into equation (\ref{equilibria2}) and solve for $P^*$. This results in the following quadratic equation in $P^*$: 

 \begin{equation}\label{EQPSTAR}
     d_1k_2P^{*2} + \left(d_1 +a_1k_2 \right)P^* + a_1 - e_0\left( \frac{a_2-d_3I^*}{d_2}\right)^\frac{1}{r}=0.
 \end{equation}

 By the Descartes' rule of signs, there is only one positive root to equation (\ref{EQPSTAR}) given by 

 \begin{equation}\label{PSTAR}
    P^*=\frac{-\left(d_1 +a_1k_2 \right) + \sqrt{a_1^2 k_2^2 +d_1^2 - 2d_1k_2 \left(a_1-2e_0 \left(\frac{a_2-d_3I^*}{d_2} \right)^\frac{1}{r} \right)}}{2d_1k_2}.
 \end{equation}
 For $P^*>0$, we require that $I^*<\frac{1}{d_3}\left(a_2-d_2\left(\frac{a_1}{e_0} \right)^r \right)$.
 To find $I^*$ explicitly, we substitute $P^*$ in equation (\ref{PSTAR}) and $S^*$ in equation (\ref{SSTAR}) into equation (\ref{equilibria1}). Notably, it is difficult to find the explicit form of $I^*$ due to the algebraic complexity of the resulting equation. However, we provide a sufficient condition for the feasibility of $I^*$ which is 
 \begin{equation}
     0<I^*<\text{min} \Biggl\{\frac{a_2}{d_3},\frac{1}{d_3}\left(a_2-d_2\left(\frac{a_1}{e_0} \right)^r \right)\Biggl\}=\frac{1}{d_3}\left(a_2-d_2\left(\frac{a_1}{e_0} \right)^r \right).
 \end{equation}
 We provide a numerical example to confirm the existence of multiple coexistence equilibria for model (\ref{Mainsystem}). See Figure \ref{fig:nullclines}.
 \end{enumerate}

 \begin{figure}[tbh]
\centering
\includegraphics[scale=0.5]{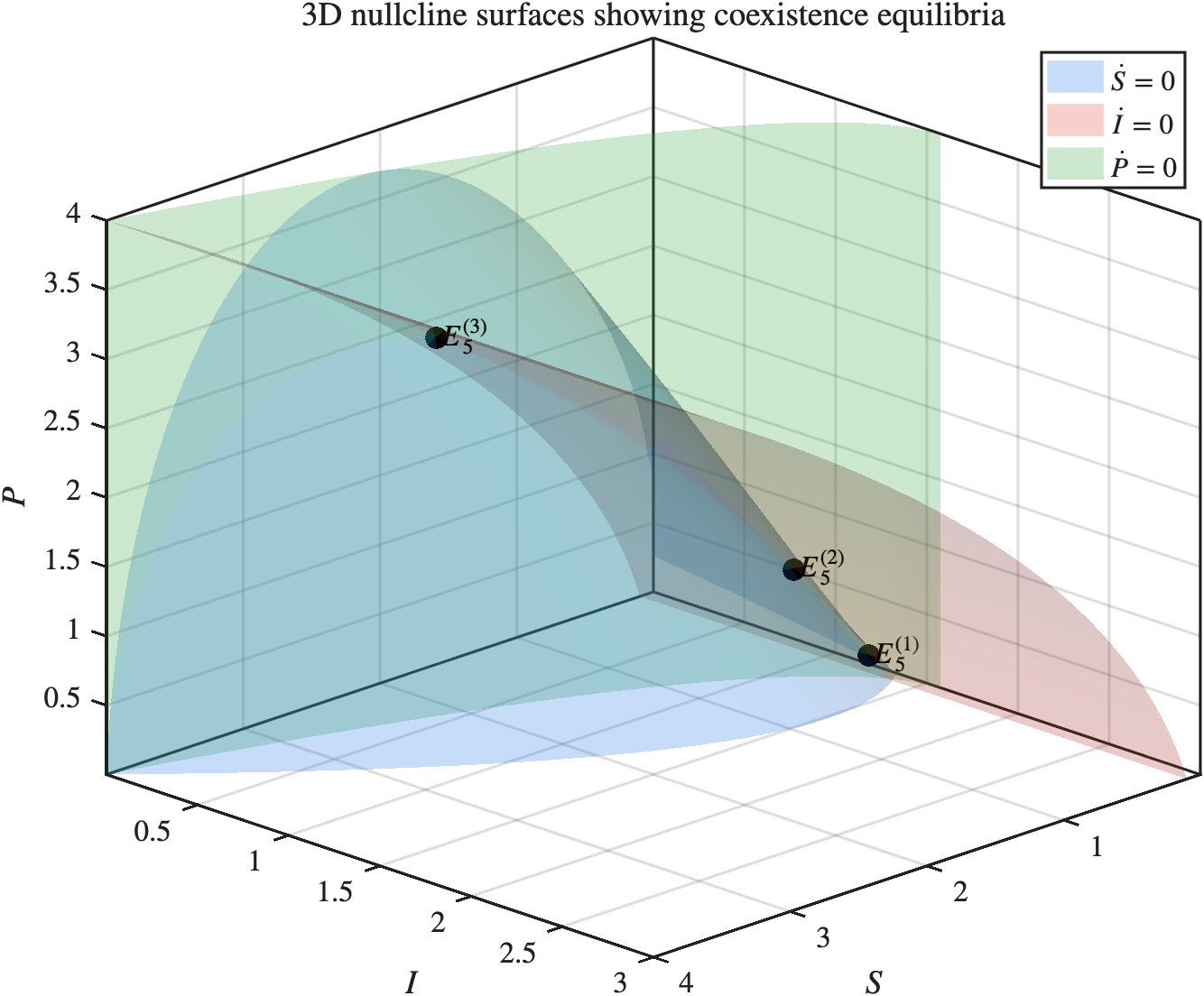}\\
\caption{Three-dimensional nullcline surfaces of model \eqref{Mainsystem} in the \((S,I,P)\)-phase space.
Paramter values are taken from Table \ref{tab:description} with $k_1=1.5$. The black markers denote the three coexistence equilibria obtained numerically: $E_5^{(1)}=(0.1472,1.2931,0.1561)$, $E_5^{(2)}=(0.4192,1.0820,0.7768)$ and $E_5^{(3)}=(2.1543,0.4258,2.7294)$}
\label{fig:nullclines}
\end{figure}

 \subsection{Local stability analysis}
The Jacobian of model (\ref{Mainsystem}) for any nonnegative equilibrium point $E(S^*,I^*,P^*)$ is computed as

\begin{equation}\label{Jaco}
   J_E= \left(
\begin{array}{ccc}
 J_{11} & J_{12} & J_{13} \\
 J_{21}& J_{22} & J_{23} \\
J_{31} & J_{32} & J_{33} \\
\end{array}
\right)
\end{equation}

where 
\begin{align*}
    J_{11}&=\frac{b_0 \left(I^* (L-2 S^*)+2 S^* (K+L)-K L-3 S^{*2}\right)}{K \left(1+k_1P^* \right)} - d_0r P^* S^{*\left(r-1\right)}-\frac{e_0 I^*}{1+k_2 P^*},\\
    J_{12}&= \frac{b_0 S^* (L-S^*)}{K \left(1+k_1 P^* \right)}-\frac{e_0 S^*}{1+k_2 P^*},\\
    J_{13}&=\frac{b_0 k_1 S^* (S^*-L) (I^*-K+S^*)}{K \left(1+k_1 P^*\right){}^2}-d_0 S^{*r}+\frac{ k_2 e_0 I^*S^*}{\left(1+k_2 P^*\right){}^2},\\
    J_{21}&= \frac{e_0I^*}{1+k_2 P^*},\\
    J_{22}&=-a_1-d_1 P^*+\frac{e_0 S^*}{1+k_2 P^*},\\
    J_{23}&= -d_1 I^*-\frac{ k_2 e_0 I^* S^*}{\left(1+k_2 P^*\right){}^2},\\
    J_{31}&=d_2 rP^* S^{*\left(r-1\right)},\\
    J_{32}&=d_3 P^*,\\
    J_{33}&=-a_2+d_3 I^*+d_2 S^{*r}.\\
\end{align*}

\begin{theorem}\label{E1}
  The susceptible prey only equilibrium  $E_1$ is locally stable if  $L<K<\frac{a_1}{e_0}$ and $r<\frac{\text{ln}~\left(\frac{a_2}{d_2} \right)}{\text{ln}~K}$.
\end{theorem}

\begin{proof}
Suppose $L<K<\frac{a_1}{e_0}$ and $r<\frac{\text{ln}~\left(\frac{a_2}{d_2} \right)}{\text{ln}~K}$. Then an evaluation of the Jacobian at $E_1$ yields
    \begin{equation}
        J_{E_1}=\left(
\begin{array}{ccc}
 b_0 (L-K) & b_0 (L-K)-K e_0 & -d_0 K^r \\
 0 & K e_0-a_1 & 0 \\
 0 & 0 & d_2 K^r-a_2 \\
\end{array}
\right).
    \end{equation}
The eigenvalues associated with $J_{E_1}$ are $\lambda_1=-b_0 (K-L)$, $\lambda_2=d_2 K^r-a_2$ and $\lambda_3=Ke_0-a_1$ are all negative. Hence $E_1$ is locally stable. However, $E_1$ is locally unstable if at least one of the following conditions are met: $L>K$ or $K>\frac{a_1}{e_0}$ or $r>\frac{\text{ln}~\left(\frac{a_2}{d_2} \right)}{\text{ln}~K}$.
\end{proof}

\begin{theorem}\label{E2}
  The susceptible prey only equilibrium  $E_2$ is locally stable if  $K<L<\frac{a_1}{e_0}$ and $r<\frac{\text{ln}~\left(\frac{a_2}{d_2} \right)}{\text{ln}~L}$.   
\end{theorem}

\begin{proof}
 Assume $K<L<\frac{a_1}{e_0}$ and $r<\frac{\text{ln}~\left(\frac{a_2}{d_2} \right)}{\text{ln}~L}$. The Jacobian matrix evaluated at $E_2$ is given by   
\begin{equation}
   J_{E_2}= \left(
\begin{array}{ccc}
 \frac{b_0 L (K-L)}{K} & -L e_0 & -d_0 L^r \\
 0 & L e_0-a_1 & 0 \\
 0 & 0 & d_2 L^r-a_2 \\
\end{array}
\right).
\end{equation}
The eigenvalues of $J_{E_2}$ are $\lambda_1=\frac{b_0 L (K-L)}{K}$, $\lambda_2=d_2 L^r-a_2$ and $\lambda_3=L e_0-a_1$. Based on earlier assumptions, each eigenvalue is negative and hence $E_2$ is locally stable. This susceptible prey only state loses its stability when either  $K>L$ or $L>\frac{a_1}{e_0}$ or $r>\frac{\text{ln}~\left(\frac{a_2}{d_2} \right)}{\text{ln}~L}$.
\end{proof}

With regards to the stability of $E_3$, an evaluation of $J_{E_3}$ results in the matrix

\begin{equation}
  J_{E_3} =  \left(
\begin{array}{ccc}
 A_{11} & A_{12} & A_{13} \\
 0 & A_{22} & 0 \\
A_{31} & A_{32} & 0\\
\end{array}
\right)
\end{equation}

where
\begin{align*}
    A_{11}&= \frac{b_0 \left(2 (K+L) \left(\frac{a_2}{d_2}\right)^{1/r}-3 \left(\frac{a_2}{d_2}\right)^{2/r}-K L\right)}{K \left(1+k_1P_3 \right)}- d_0 r \left(\left(\frac{a_2}{d_2}\right)^{1/r}\right)^{r-1}P_3, \\
    A_{12}&=\left(\frac{a_2}{d_2}\right)^{1/r} \left(\frac{b_0 \left(L-\left(\frac{a_2}{d_2}\right)^{1/r}\right)}{K\left(1+k_1 P_3 \right)}-\frac{e_0}{1+ k_2 P_3}\right), \\
    A_{13}&=\frac{b_0 k_1 \left(\frac{a_2}{d_2}\right)^{1/r} \left(\left(\frac{a_2}{d_2}\right)^{1/r}-K\right) \left(\left(\frac{a_2}{d_2}\right)^{1/r}-L\right)}{K \left(1+k_1P_3 \right)^2} - \frac{d_0a_2}{d_2}, \\
    A_{22}&= \frac{e_0 \left(\frac{a_2}{d_2}\right)^{1/r}}{1+k_2P_3}-a_1- d_1 P_3, \\
    A_{31}&=  d_2 r \left(\left(\frac{a_2}{d_2}\right)^{1/r}\right)^{r-1} P_3, \\
    A_{32}&= d_3 P_3. \\
\end{align*}
One of the eigenvalues of $ J_{E_3}$ is $\lambda_1=A_{22}$ and the other two eigenvalues are roots of the characteristic equation 
\begin{equation}\label{char1}
    \lambda^2 -A_{11}\lambda - A_{13} A_{31}=0 
\end{equation}

which satisfy $\lambda_2+\lambda_3=A_{11}$ and $\lambda_2\lambda_3= - A_{31} A_{13}$. By the Routh-Hurwitz stability criteria, the infectious prey-free state is locally stable provided $A_{22}<0$, $A_{11}<0$ and $-A_{31}A_{13}>0$. (This happens only if $A_{13}<0$). Next, we state a theorem concerning the local stability of the infectious prey-free state.

\begin{theorem}\label{E3}
  The infectious prey-free equilibrium $E_3$ is locally stable if  $A_{22}<0$, $A_{11}<0$ and $A_{13}<0$.
\end{theorem}

Next, an evaluation of the Jacobian at $E_4$ yields
\begin{equation}\label{JacE4}
   J_{E_4}= \left(
\begin{array}{ccc}
 B_{11} & B_{12} & B_{13} \\
 B_{21} & B_{22} & B_{23} \\
 0 & 0 & B_{33}\\
\end{array}
\right)
\end{equation}
where

\begin{align*}
    B_{11}&=\frac{b_0 \left(2 a_0 e_0 (-I_4+K+L)-3 a_0^2+L e_0^2 (I_4-K)\right)}{K e_0^2}-e_0 I_4, \\
    B_{12}&=\left(\frac{b_0 \left(L e_0-a_0\right)}{K e_0^2}-1\right),\\
    B_{13}&=\frac{a_0 b_0 k_1 \left(a_0-L e_0\right) \left(a_0+e_0 (I_4-K)\right)}{K e_0^3}-d_0 \left(\frac{a_0}{e_0}\right)^r+a_0 k_2 I_4,  \\
    B_{21}&= e_0 I_4,\\
    B_{22}&=a_0-a_1, \\
    B_{23}&=- \left(a_0 k_2+d_1\right)I_4,\\
    B_{33}&=d_2 \left(\frac{a_0}{e_0}\right)^r-a_2+d_3 I_4.\\
\end{align*}

The characteristic equation for $J_{E_4}$ is 
\begin{equation}
    \lambda^3 + \zeta_1 \lambda^2 +\zeta_2 \lambda + \zeta_3=0
\end{equation}
where
\begin{align*}
    \zeta_1 &= -\left(B_{11}+B_{22} \right),\\
    \zeta_2 & = B_{11} B_{22} - B_{12}B_{21}, \\
    \zeta_3 & = B_{12} B_{21} B_{33} - B_{11} B_{22} B_{33}.
\end{align*}

We can ascertain the local stability of $E_4$ by applying the Routh-Hurwitz stability criteria in the following theorem:
\begin{theorem}\label{E4}
  The predator-free equilibrium $E_4$ is locally stable if  $\zeta_1 >0, \zeta_2>0, \zeta_3>0$ and $\zeta_1 \zeta_2 -\zeta_3>0$.
\end{theorem}

Finally, we evaluate $J_{E_5}$ and obtain the following characteristic equation: 
\begin{equation}\label{charsigma}
    \lambda^3 + \sigma_1 \lambda^2 +\sigma_2 \lambda + \sigma_3=0
\end{equation}

where

\begin{align*}
    \sigma_1 &= -\left(J_{11}+J_{22} +J_{33} \right),\\
    \sigma_2 &= J_{22}J_{33}+J_{11}\left(J_{22}+J_{33} \right)-J_{12}J_{21}-J_{13}J_{31}-J_{23}J_{32},\\
    \sigma_3 &=J_{13}J_{22}J_{31} - J_{12}J_{23}J_{31}-J_{13}J_{21}J_{32} +J_{11}J_{23}J_{32}+J_{12}J_{21}J_{33}-J_{11}J_{22}J_{33}.\\
\end{align*}

Similarly, we apply the Routh-Hurwitz stability criteria for local stability of the coexistence equilibrium via the next theorem.

\begin{theorem}\label{E5}
  The coexistence equilibrium $E_5$ is locally stable if  $\sigma_1 >0, \sigma_2>0, \sigma_3>0$ and $\sigma_1 \sigma_2 -\sigma_3>0$.
\end{theorem}

\subsection{On the extinction equilibrium}
The extinction equilibrium $E_0$ requires special consideration because the aggregation term $S^r$, with $0<r<1$, is not differentiable at $S=0$. Consequently, the Jacobian matrix is not defined at $E_0$, and the classical Hartman--Grobman linearization theorem cannot be applied. The local dynamics near $E_0$ must therefore be investigated using nonlinear techniques rather than standard linearization arguments.

A complete characterization of the stability of $E_0$ for the full three-dimensional system is challenging due to the non-smooth nature of the aggregation term. Therefore, to gain analytical insight into the role of the Allee threshold in shaping the extinction dynamics, we restrict our attention to the biologically relevant invariant susceptible-only manifold
$$
\mathcal{A}=\{(S,I,P)\in\mathbb{R}_+^3: S\geq 0, I=0, P=0\}.
$$
This reduction isolates the demographic dynamics of the susceptible prey population and allows us to rigorously examine the local behavior of solutions near the extinction state.

\begin{proposition}
\label{prop:E0unstable}
Suppose \(L<0\). Then the extinction equilibrium $E_0$ is unstable relative to the invariant manifold $\mathcal{A}$. In particular, $E_0$ is not assymptotically stable in the full model.
\end{proposition}

\begin{proof}
On $\mathcal A$, model \eqref{Mainsystem} reduces to

$$
\frac{dS}{dt}
=
b_0S\left(1-\frac{S}{K}\right)(S-L).
$$

Since $L<0$, we have $S-L>0$ for all sufficiently
small $S>0$. Moreover,
$$
1-\frac{S}{K}>0.
$$

Hence
$$
\frac{dS}{dt}>0
$$

for all sufficiently small $S>0$.

\noindent Therefore, solutions starting arbitrarily close to $E_0$
along $\mathcal A$ move away from the origin. Consequently $E_0$ is unstable.
\end{proof}

\begin{theorem}\label{thm:S00invariant}
Suppose $L>0$ and consider the invariant manifold $\mathcal{A}$. Then the extinction equilibrium $E_0$ is locally asymptotically stable relative to $\mathcal{A}$. More precisely, there exists $\delta \in (0,L)$ such that every solution with
\[
0<S(0)<\delta
\]
satisfies
\[
\lim_{t\to\infty}S(t)=0.
\]
\end{theorem}

\begin{proof}
Restrict model \eqref{Mainsystem} to the invariant manifold
$\mathcal{A}$.
Then

\[
\frac{dS}{dt}
=
b_0S\left(1-\frac{S}{K}\right)(S-L).
\]

Choose $\delta=L/2$.
For $0<S<\delta$ we have

\[
S-L<-\frac{L}{2}<0
\]

and

\[
1-\frac{S}{K}>0.
\]

Therefore

\[
\frac{dS}{dt}
=
b_0S\left(1-\frac{S}{K}\right)(S-L)
<
-\frac{b_0L}{2}
\left(1-\frac{\delta}{K}\right)S.
\]

Define

\[
c=
\frac{b_0L}{2}
\left(1-\frac{\delta}{K}\right)>0.
\]

Then

\[
\frac{dS}{dt}\le -cS.
\]

Applying the comparison theorem yields

\[
S(t)
\le
S(0)e^{-ct}.
\]

Hence

\[
\lim_{t\to\infty}S(t)=0.
\]

\noindent Therefore $E_0$ is locally asymptotically stable relative to the susceptible-only manifold.
\end{proof}

\begin{figure}[tbh]
\centering
\includegraphics[width=1.0\textwidth]{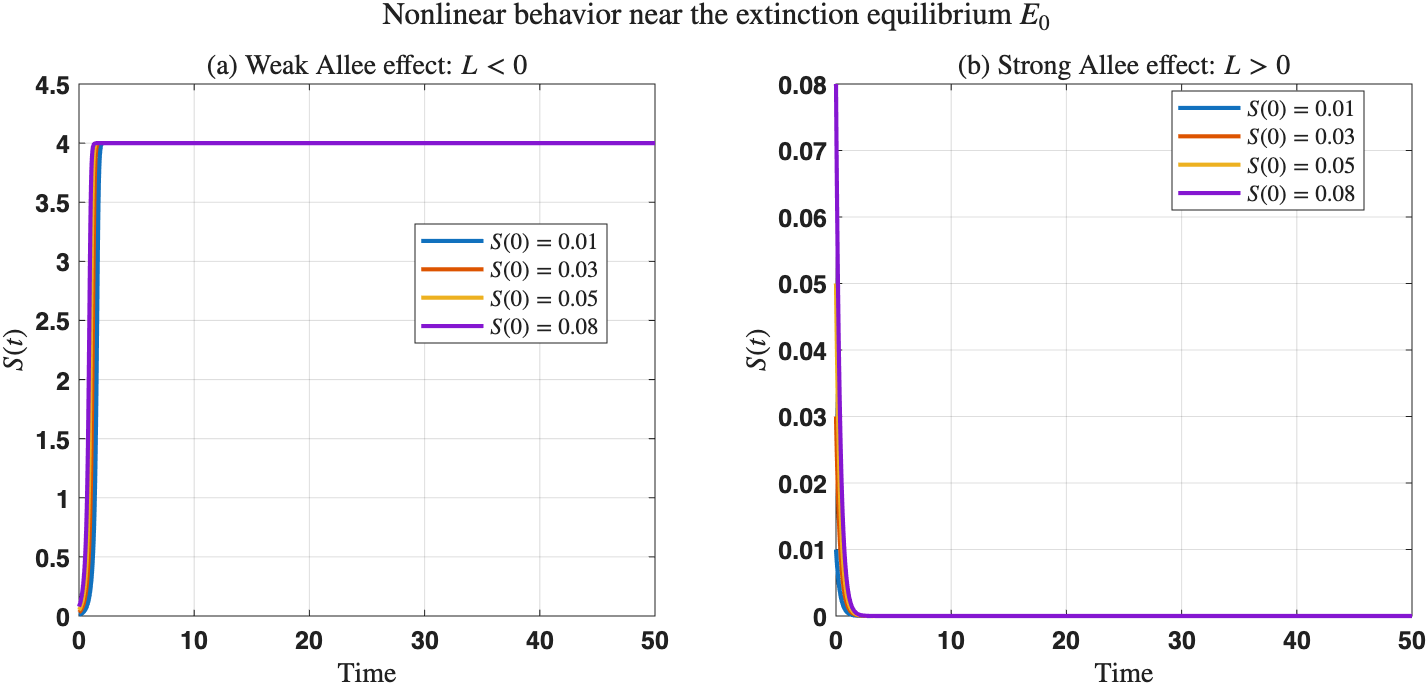}
\caption{Numerical behavior near the extinction equilibrium $(E_0)$ on the susceptible-only invariant manifold. In panel (a), $L=-1<0$, and trajectories starting near $E_0$ move away from extinction, confirming instability under a weak Allee effect. In panel (b), $L=1>0$, and trajectories starting below the Allee threshold converge toward $E_0$, confirming local attraction under a strong Allee effect.}
\label{fig:E0fig}
\end{figure}

\begin{remark}
Theorem \ref{thm:S00invariant} highlights a fundamental distinction between weak and strong Allee effects. For $L<0$, Proposition \ref{prop:E0unstable} shows that the extinction state is unstable because low-density populations retain positive per-capita growth (see Figure \ref{fig:E0fig}(a) for numerical validation). In contrast, when $L>0$, populations below the critical threshold experience negative growth and are driven toward extinction (see Figure \ref{fig:E0fig}(b) for numerical validation). Thus the Allee threshold transforms the extinction equilibrium from a repelling state into a locally attracting state on biologically relevant low-density trajectories. Although Proposition \ref{prop:E0unstable} shows that the extinction
equilibrium is unstable when $L<0$, this does not preclude extinction
of the susceptible prey population. In particular, the aggregation-based
predation term
$$
d_0S^rP,
\qquad 0<r<1,
$$
\noindent introduces a non-Lipschitz mechanism that can dominate the growth term when predator pressure remains sufficiently large. This phenomenon highlights a fundamental distinction between local stability and extinction dynamics in systems with aggregation-mediated predation.
\end{remark}

\section{Bifurcation Dynamics and Ecological Tipping Points}\label{sec:Bifurcation analysis}
\noindent In this section, we explore and analyze the occurrence of possible local co-dimension one and two bifurcations. These bifurcations provide valuable insight into the qualitative changes in the dynamics of model (\ref{Mainsystem}) that arise when key parameters of interest - namely the level of fear $k_1$, Allee effect $L$ and prey aggregation constant $r$ are varied.

\subsection{Co-dimension one bifurcation}
Co-dimension one bifurcations are typically associated with the creation, collision of steady states or period orbits when single parameters of interest are varied. We establish the conditions for the occurrence of a saddle-node bifurcation, Hopf bifurcation and transcritical bifurcation via the following theorems:

\begin{theorem}\label{SNB}
The model (\ref{Mainsystem}) experiences a saddle-node bifurcation around $E_5$  at $L=L^*$ when $tr(J_{E_5})<0$ and $det(J_{E_5})=0$.
\end{theorem}

\begin{proof}
Suppose $tr(J_{E_5})<0$ and $det(J_{E_5})=0$. Then model (\ref{Mainsystem}) has a simple zero eigenvalue. Let $v=(v_1,v_2,v_3)^T\neq 0$ and $w=(w_1,w_2,w_3)^T\neq 0$ be eigenvectors corresponding the zero eigenvalue of $J_{E_5}$ and $J_{E_5}^T$ respectively. Let $Z=(z_1,z_2,z_3)^T$, where 
\begin{align*}
z_1 &= \frac{b_0 S}{1 + k_1 P} \left(1 - \frac{S + I}{K} \right)(S - L) - d_0 S^r P - \frac{e_0 S I}{1 + k_2 P}, \\
   z_2 &= \frac{e_0 S I}{1 + k_2 P} - d_1 I P - a_1 I, \label{Mainsystem} \\
z_3 &= d_2 S^r P + d_3 I P - a_2 P.
\end{align*}
By ensuring the $w_1 \neq 0$ and $S^*+I^* \neq K$,
\begin{align*}
    w^TZ_L \left(E_2,L^* \right) &= w^T\left(-\frac{b_0S^*}{1+k_1P^*}\left(1-\frac{S^*+I^*}{K} \right),0,0\right)^T\\
    &=- \frac{w_1b_0S^*}{1+k_1P^*}\left(1-\frac{S^*+I^*}{K} \right)\\
    & \neq 0.
\end{align*}
Therefore the second transversality condition is validated. Furthermore, 
\begin{align*}
   w^T\left[D^2 Z \left(E_5,L^* \right)\left(v,v\right) \right] & \neq 0,   
\end{align*}
where $D^2 Z = \left(\nabla^2 z_1,\nabla^2 z_2,\nabla^2 z_3 \right)^2$. By Sotomayor's theorem \cite{Perko13}, model (\ref{Mainsystem}) undergoes a saddle-node bifurcation around the coexistence equilibrium $E_5$ at $L=L^*$.
\end{proof}

\begin{theorem}\label{HB}
Model (\ref{Mainsystem}) experiences a Hopf bifurcation around $E_5$  at $r=r^*$ provided the following conditions are satisfied:
\begin{eqnarray}\label{c1}
		\sigma_{1}(r_{H})>0, ~\sigma_{3}(r_{H})>0,~ \sigma_{1}(r_{H})\sigma_{2}(r_{H}) -\sigma_{3}(r_{H})=0
		\end{eqnarray}
		and
		\begin{eqnarray}\label{c2}
		\left[\sigma_{1}(r)\sigma_{2}(r)\right]^{\prime}_{r = r_{H}} - \sigma_{3}^{\prime}(r_{H}) \neq 0.
		\end{eqnarray}
\end{theorem}

\begin{proof}
We rewrite the characteristic equation (\ref{charsigma}) in the form
		\begin{equation}\label{characeq}
		\left[\lambda^{2}(r_{H}) + \sigma_{2}(r_{H})\right] \left[\lambda(r_{H})+\sigma_{1}(r_{H})\right]=0 ,
		\end{equation}
so that a Hopf bifurcation can occur with roots $\lambda_{1}(r_{H}) = i \sqrt{\sigma_{2}(r_{H})},$  $\lambda_{2}(r_{H}) = -i \sqrt{\sigma_{2}(r_{H})},$  $\lambda_{3}(r_{H}) = - \sigma_{1}(r_{H})<0$. It is easy to see that $\sigma_3(r_{H}) = \sigma_1(r_{H})\sigma_2(r_{H})$. Next, we establish the transversality condition 
		\begin{equation}
		\dfrac{d(Re\lambda_{k}(r))}{dr}\Bigg|_{r=r_{H}}\not=0, k=1,2,
		\end{equation} 
to enable the  verification of the existence of periodic solutions bifurcating around $E_5^*$ at $r =r_{H}$.
We substitute $\lambda_{k}(r) = \Gamma(r)+i\Lambda(r)$ into (\ref{characeq}) and compute the derivative to obtain
		
	\begin{eqnarray}
		\varpi_{1}(c)\Gamma^{\prime}(c)-\varpi_{2}(c)\Lambda^{\prime}(c) +\varpi_{4}(c) &=& 0, \label{u1}\\
		\varpi_{2}(c)\Gamma^{\prime}(c) + \varpi_{1}(c)\Lambda^{\prime}(c) + \varpi_{3}(c)&=&0,\label{u2}
		\end{eqnarray}
		where
		\begin{eqnarray*}
			\varpi_{1}(r)&=&3\Gamma^{2}(r)-3\Lambda^{2}(r)+\sigma_{2}(r)+2\sigma_{1}(r)\Gamma(r),\\
			\varpi_{2}(r)&=& 6\Gamma(r)\Lambda(r)+2\sigma_{1}(r)\Lambda(r),\\
			\varpi_{3}(r)&=&2\Gamma(r)\Lambda(r)\sigma_{1}^{\prime}(r)+\sigma_{2}^{\prime}(r)\Lambda(r),\\
			\varpi_{4}(r)&=&\sigma_{2}^{\prime}(r)\Gamma(r)+\Gamma^{2}(r)\sigma_{1}^{\prime}(r)-\Lambda^{2}(r)\sigma_{1}^{\prime}(r) + \sigma_{3}^{\prime}(r).
		\end{eqnarray*}
		
By using Cramer's rule, we solve for $\Gamma^{\prime}(r_{H})$ from the linear systems in \eqref{u1} and \eqref{u2}. It is worth noting that,				
\noindent at $r=r_{H},$ $\Gamma(r_{H})=0$ and $\Lambda(r_{H})=\sqrt{\sigma_{2}(r_{H})}$, resulting in
\begin{eqnarray*}
			\varpi_{1}(r_{H}) &=& -2 \sigma_{2}(r_{H}),\\
			\varpi_{2}(r_{H}) &=& 2 \sigma_{1}(r_{H})\sqrt{\sigma_{2}(r_{H})},\\
			\varpi_{3}(r_{H}) &=& \sigma_{2}^{\prime}(r_{H})\sqrt{\sigma_{2}(r_{H})},\\
			\varpi_{4}(r_{H}) &=& \sigma_{3}^{\prime}(r_{H}) - \sigma_{2}(r_{H})\sigma_{1}^{\prime}(r_{H}).
		\end{eqnarray*}	
		
We now have 
		\begin{eqnarray*}
			\dfrac{dRe(\lambda_{k}(r))}{dr}\Bigg|_{r=r_{H}}&=&\Gamma^{\prime}(r_{H}),\\&=&-\dfrac{\varpi_{3}(r_{H})\varpi_{2}(r_{H})+\varpi_{4}(r_{H})\varpi_{1}(r_{H})}{\varpi_{1}^{2}(r_{H})+\varpi_{2}^{2}(r_{H})},\\
			&=&\dfrac{\sigma_{3}^{\prime}(r_{H})-\sigma_{1}(r_{H})\sigma_{2}^{\prime}(r_{H})-\sigma_{2}(r_{H})\sigma_{1}^{\prime}(r_{H})}{2\left(\sigma_{2}(r_{H})+\sigma_{1}^{2}(r_{H})\right)}\not=0,
		\end{eqnarray*}
				
\noindent on condition that $\left[\sigma_{1}(r)\sigma_{2}(r)\right]^{\prime}_{r = r_{H}} - \sigma_{3}^{\prime}(r_{H}) \neq 0.$ \\ Therefore, the transversality condition is established, which leads to the conclusion that the model experiences a Hopf bifurcation around $E_5^*$ at $r=r_{H}.$		
\end{proof}

\begin{theorem}\label{TB}
The model (\ref{Mainsystem}) experiences a transcritical bifurcation around $E_2$ at $r=r^*=\frac{\text{ln}~\left(\frac{a_2}{d_2} \right)}{\text{ln}~L}$.
\end{theorem}

\begin{proof}
Simple computations show that at the critical susceptible prey aggregation constant $r=r^*=\frac{\text{ln}~\left(\frac{a_2}{d_2} \right)}{\text{ln}~L}$ around $E_2$,
\begin{equation}\label{Trans}
   J_{E_2}= \left(
\begin{array}{ccc}
 \frac{b_0 L (K-L)}{K} & -L e_0 & -\frac{d_0 a_2}{d_2} \\
 0 & L e_0-a_1 & 0 \\
 0 & 0 & 0 \\
\end{array}
\right).
\end{equation}
The associated eigenvalues for the matrix in (\ref{Trans}) are $\lambda_1=0$, $\lambda_2=\frac{b_0 L (K-L)}{K}$ and $\lambda_3=Le_0-a_1$. We let $G=(g_1,g_2,g_3)^T$ and $N=(n_1,n_2,n_3)^T$ represent the eigenvectors corresponding to the zero eigenvalue of $J_{E_2}$ and $J_{E_2}^T$ respectively. Computations yield $G=\left(\frac{KL^{r-1}d_0}{\left(K-L\right)b_0},0,1\right)^T$ and $N=(0,0,1)^T$. 
We proceed to validate the transversality conditions using Sotomayor's theorem \cite{Perko13}. 
\begin{equation*}
    N^TZ_r \left(E_2,r^* \right) = (0,0,1)(0,0,0)^T=0.
\end{equation*}
Also,
\begin{align*}
   N^T\left[D Z_r \left(E_2,r^* \right)G \right] &=  \left(
\begin{array}{ccc}
 0 & 0 &1 \\
\end{array}
\right) 
\left(
\begin{array}{ccc}
 0& 0 & -d_0 r L^{r-1}  \\
  0 & 0 &0\\
    0 & 0 &d_2 rL^{r-1}\\
\end{array}
\right)
\left(
\begin{array}{ccc}
 \frac{KL^{r-1}d_0}{\left(K-L\right)b_0} \\
  0 \\
  1\\
\end{array}
\right) \\
&=d_2 r L^{r-1}\neq 0,
\end{align*}
and 
\begin{align*}
   N^T\left[D^2 Z \left(E_2,r^* \right)\left(G,G\right) \right] & \neq 0.   
\end{align*}
Hence the model (\ref{Mainsystem}) experiences a transcritical bifurcation at $E_2$ when $r=r^*=\frac{\text{ln}~\left(\frac{a_2}{d_2} \right)}{\text{ln}~L}$.
\end{proof}

\subsection{Numerical validation}
This subsection focuses on providing numerical experiments to validate our theoretical findings. We use MATLAB version R2026a and MATCONT \cite{Dhooge2003} software in performing the numerical experiments. We show the possible occurrence of co-dimension one bifurcations using bifurcation diagrams, in particular, for the parameters $k_1,r$ and $L$ when all other parameters are fixed. We shall use parameter values from Table \ref{tab:description} for our bifurcation experiments.

\subsubsection{Impact of Allee effect parameter $L$}
To understand the role of Allee effects in system (\ref{Mainsystem}), we fix all other parameters and vary $L$. See Figure \ref{Bif1}. The system experiences a transcritical (TC) bifurcation  at $L^{[TC]}=-3.95198$ around the predator-free state $E_4^{[TC]}=\left(0.1,1.6838,0 \right)$. The occurrence of a TC bifurcation  identifies a critical threshold for predator invasion and persistence. Ecologically, this threshold determines whether an introduced predator population can successfully establish itself or ultimately go extinct. In the context of bio-control, particularly when predators are introduced to suppress a disease-carrying prey population, the TC bifurcation provides conditions necessary for successful predator establishment and sustained regulation of the target population.  For weak Allee threshold $L\in \left(-3.95198,-2.5164 \right)$, the system experiences bistability phenomena between the predator-free state $E_4$ and the coexistence state $E_5$. Therefore depending on the choice of initial data, solutions will either be attracted to $E_4$ or $E_5$. Within this threshold interval, a Hopf bifurcation is noted at the critical value $L_1^{[H]}=-3.2192$ around $E_5^{[H]}=\left(0.1805, 1.5752, 0.2683 \right)$. The first Lyapunov coefficient is computed as $\chi_{L_1}=1.5666\times 10^1$ and hence the bifurcation is subcritical. An increase away from the weak Allee threshold interval leads to the occurrence of another Hopf bifurcation at $L_2^{[H]}=-2.5164$ around $E_4^{[H]}=\left(0.1, 1.2836, 0 \right)$ with its first Lyapunov coefficient calculated as $\chi_{L_2}=-7.7542\times 10^{-2}$ and hence the bifurcation is a supercritical Hopf. A further increase of the weak Allee threshold causes the existence of a third Hopf bifurcation at $L_3^{[H]}=-0.7626$ around $E_5^{[H]}=\left(1.4889, 0.7798, 2.3043 \right)$. A similar computation of the first Lyapunov coefficient yields $\chi_{L_3}=-2.0029 \times 10^0$ which means this bifurcation is also a supercritical Hopf. From our simulation, as the Allee threshold becomes less negative, the susceptible prey becomes vulnerable to demographic pressure or limitations which can lead to a (de)stabilization of the coexistence equilibrium and drive the system towards persistent population cycles. A final increment of the weak Allee threshold causes the occurrence of a saddle-node bifurcation at $L_4^{[SN]}=-0.4882$ around $E_5^{[SN]}=\left(2.1913,0.5197,2.9760 \right)$. Ecologically, the saddle-node bifurcation represents a catastrophic tipping point with direct conservation implications. Beyond $L_4^{[SN]}$, it will be difficult for management interventions to restore species persistence.



\begin{figure}[tbh]
\centering

\begin{minipage}{0.48\textwidth}
\centering
\includegraphics[width=\textwidth]{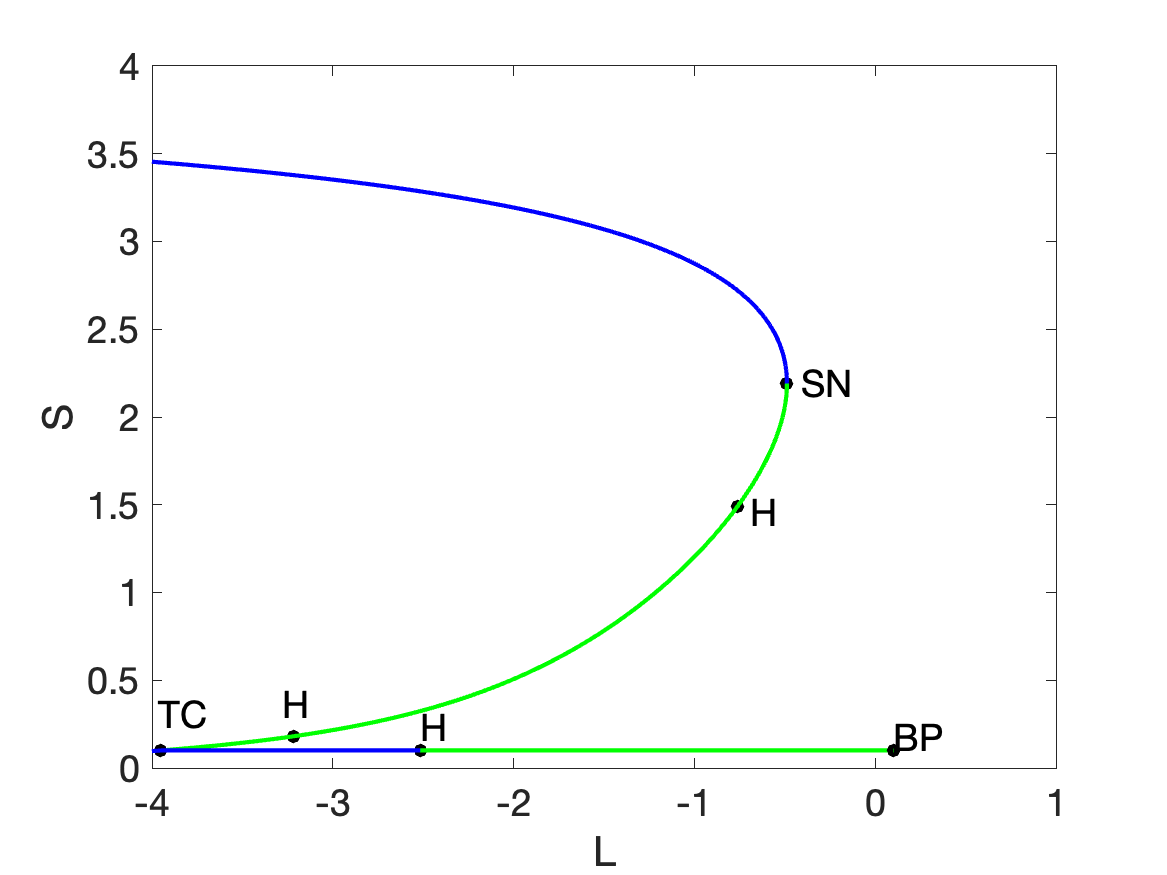}
\\[2mm]
$(a)$
\end{minipage}
\hfill
\begin{minipage}{0.48\textwidth}
\centering
\includegraphics[width=\textwidth]{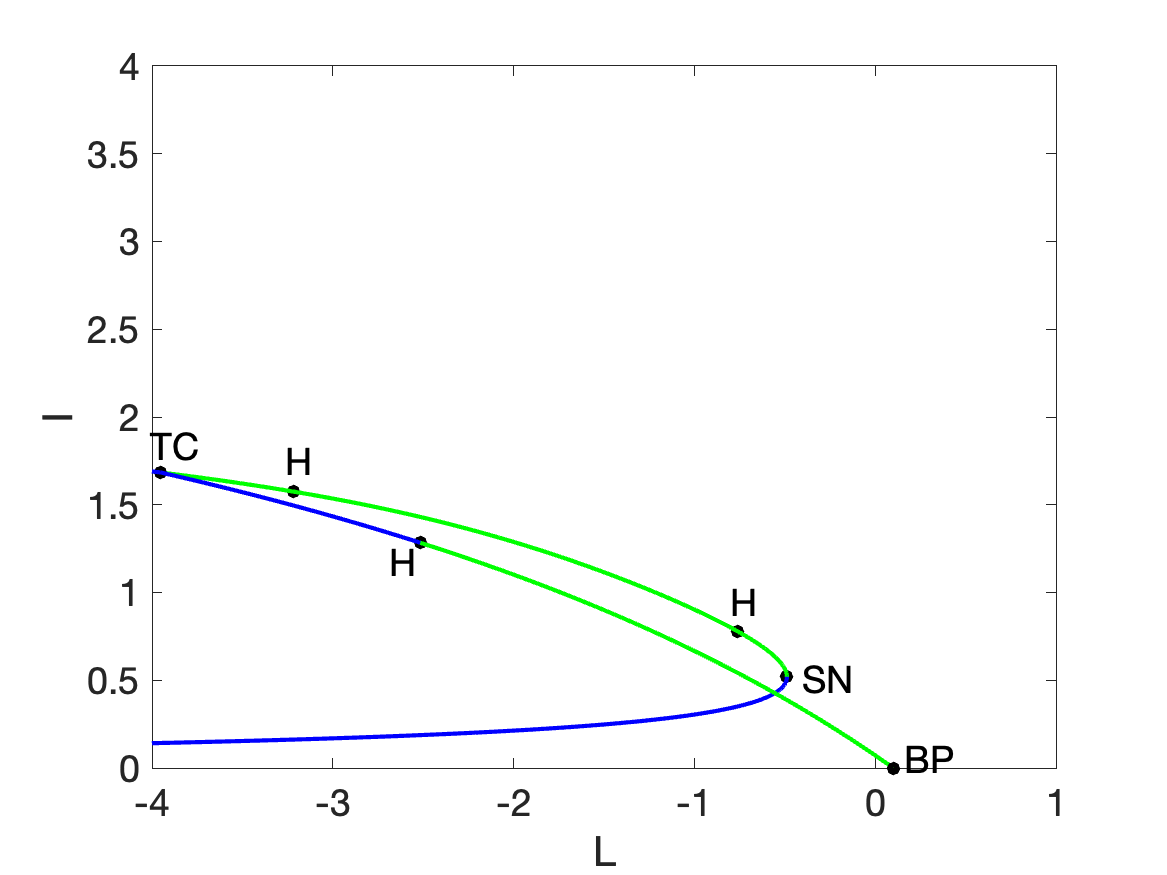}
\\[2mm]
$(b)$
\end{minipage}

\vspace{0.5cm}

\begin{minipage}{0.55\textwidth}
\centering
\includegraphics[width=\textwidth]{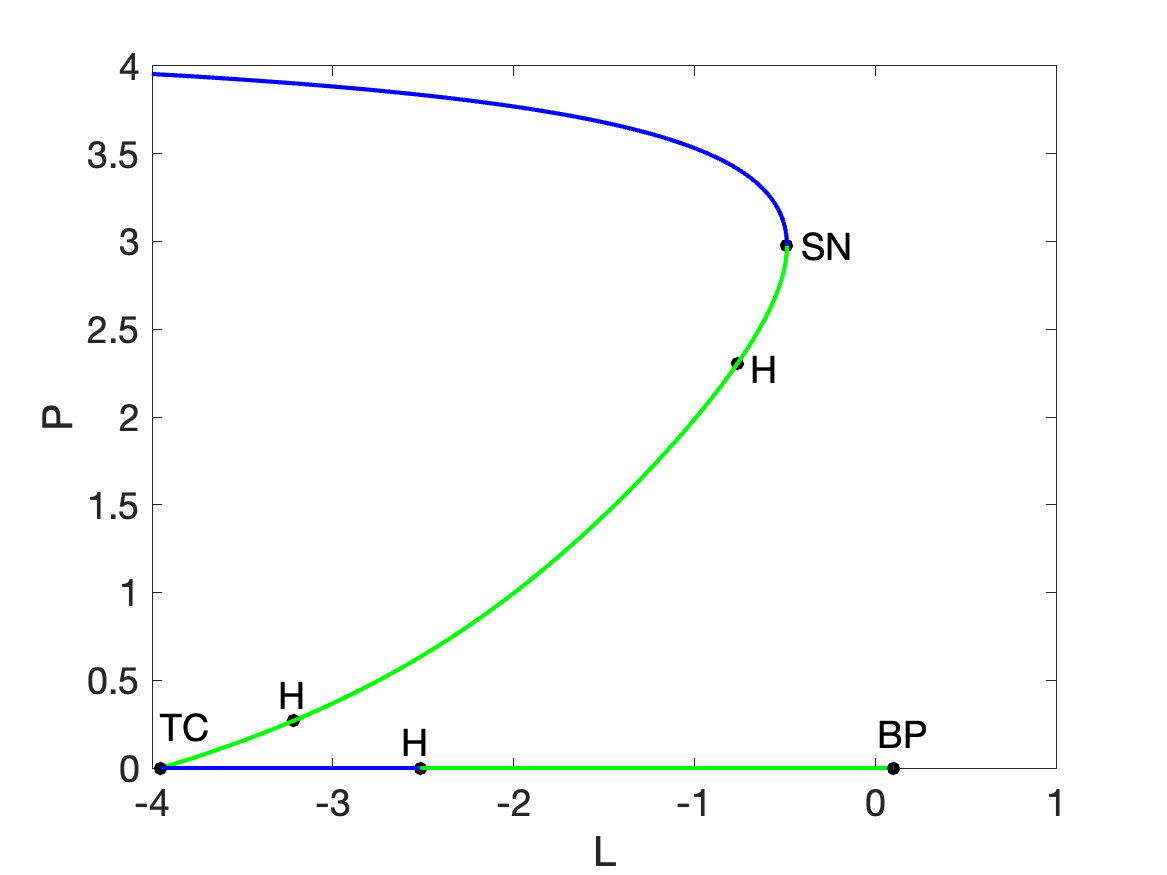}
\\[2mm]
$(c)$
\end{minipage}

\caption{Bifurcation diagram for $L$ for the model (\ref{Mainsystem}). The blue and green lines denote stable and unstable branches respectively. Note: SN=Saddle-Node point, H=Hopf point, TC = Transcritical point, BP = Branch point.}
\label{Bif1}
\end{figure}

\subsubsection{Impact of prey aggregation $r$}
From the bifurcation diagram in Figure \ref{Bif2}, if the strength of prey aggregation $r$ is the interval $(0,0.313944]$, the populations do not persist. This is due to a saddle-node bifurcation occurring at $r^{[SN]}=0.313944$ around $E_5^{[SN]}=(2.350352, 0.692280, 3.112553)$. This shows that aggregation strength must exceed a threshold for all populations to coexist. A small increase in $r$ results in the occurrence of a Hopf bifurcation at $r_1^{[H]}=0.319942$ around $E_5^{[H]}=(2.252766, 0.703272, 3.029404)$ with a first Lyapunov coefficient computed as $\chi_{r_1} =-3.198244 \times 10^1$ and hence supercritical. A further increase leads to a transcritical critical bifurcation at $r^{[TC]}=0.544980$ around the infected prey-free equilibrium $E_3^{[TC]}=(3.567540, 0, 4.032677)$. The transcritical bifurcation provides an aggregation threshold where predators indirectly can eliminate disease from an ecosystem by depleting the susceptible prey population. This reduces the pool of hosts available for disease transmission. As a result, the infected class cannot replenish itself through new infections and ultimately goes extinct. Finally, we observe a second Hopf bifurcation occurring at $r_2^{[H]}=0.719768$ around $E_5^{[H]}=(2.619583,0, 5.424286)$ with a computed Lyapunov coefficient $\chi_{r_2}=-1.024881 \times 10^0$ classifying the Hopf bifurcation as supercritical.


\begin{figure}[tbh]
\centering

\begin{minipage}{0.48\textwidth}
\centering
\includegraphics[width=\textwidth]{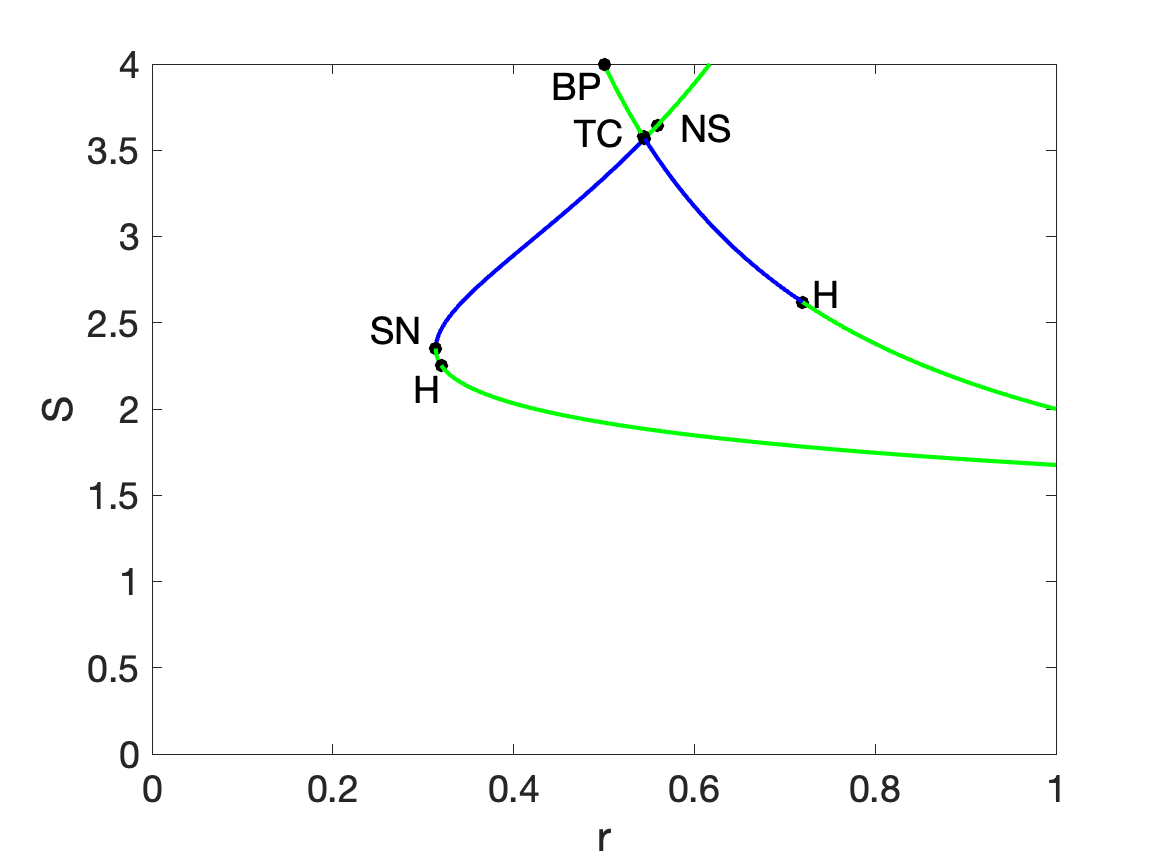}
\\[2mm]
$(a)$
\end{minipage}
\hfill
\begin{minipage}{0.48\textwidth}
\centering
\includegraphics[width=\textwidth]{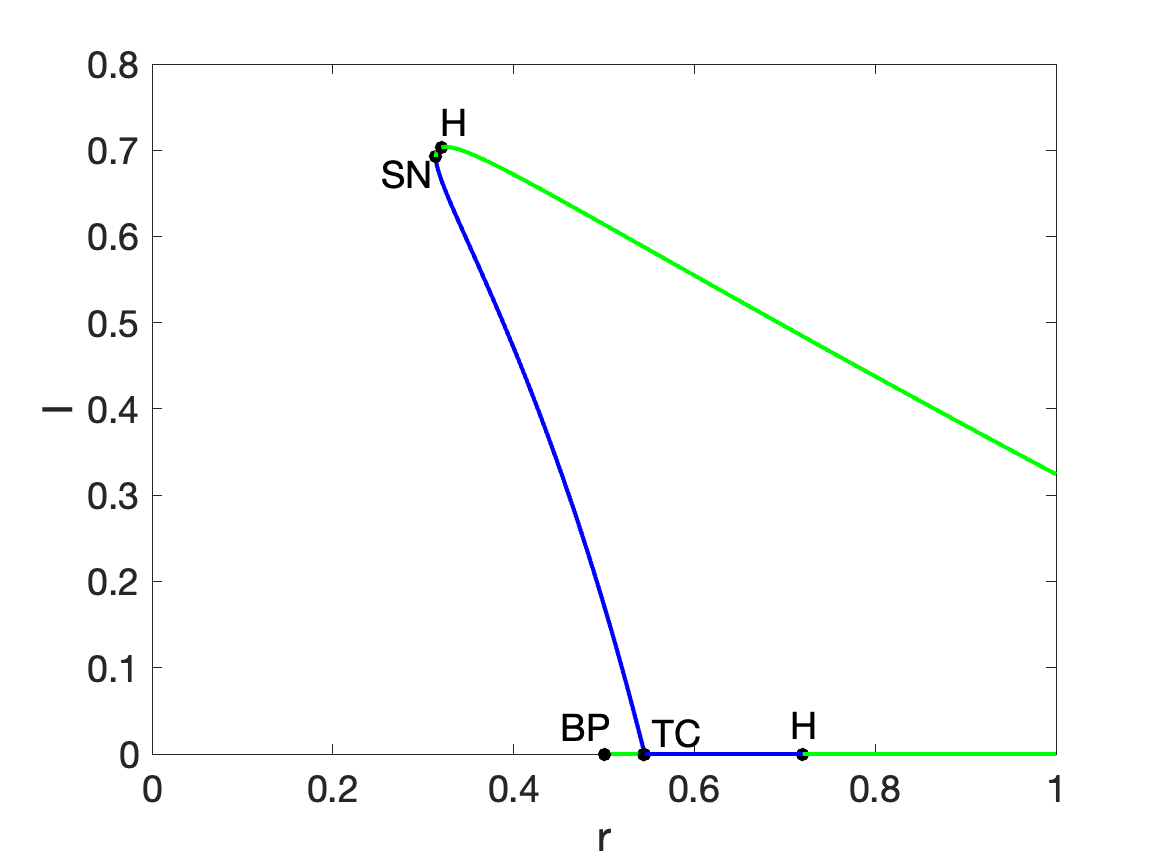}
\\[2mm]
$(b)$
\end{minipage}

\vspace{0.5cm}

\begin{minipage}{0.55\textwidth}
\centering
\includegraphics[width=\textwidth]{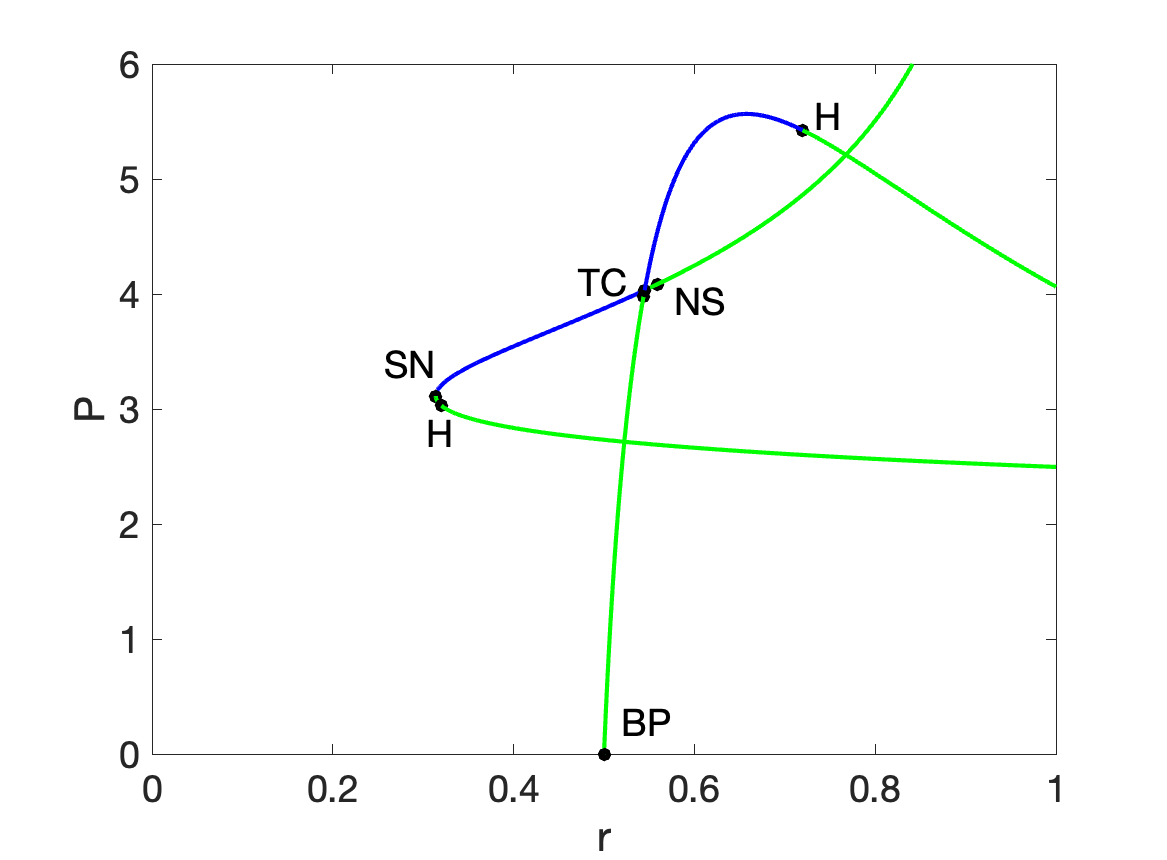}
\\[2mm]
$(c)$
\end{minipage}

\caption{Bifurcation diagram for predator induced fear $r$ on susceptible prey for model (\ref{Mainsystem}). The blue and green  lines denote stable and unstable branches respectively. In addition, we choose $b_0=8$ and $L=1$.  Note: SN=Saddle-Node point, H=Hopf point, BP =  Branch point, NS = Neutral Saddle.}
\label{Bif2}
\end{figure}

\subsubsection{Impact of fear levels $k_1$}
We explore how predator-induced fear interacts with the Allee effect and aggregation dynamics to influence the persistence of susceptible prey in Figure \ref{Hbif2}. When we set $L=1$, the Allee threshold is a strong one. For low fear levels $k_1 \in \left(0,0.364829 \right)$, two coexistence equilibria exist due to the density of the susceptible prey population not going below the strong Allee threshold of $L=1$. At $k_1=0.364829$, the system experiences a Hopf bifurcation around $E_5^{[H]}=\left(2.451314, 0.347466, 2.96307 \right)$ with its first Lyapunov coefficient calculated as $\chi_{k_1}=-2.542146$ and hence the bifurcation is supercritical. A slight further increase in fear levels results in the two coexistence equilibria colliding and annihilating each other via a saddle-node bifurcation at $k_1=0.3677359$ around $E_5^{[SN]}= \left(2.546129,0.32347,3.03465 \right)$. Similar dynamics is seen when $L=0$. However, when $L=-1$, the Allee effect becomes a weak one. The population density of the susceptible prey is negatively impacted for fear levels $k_1$ in the range $(0, 1.417511)$. When $k_1^{[SN]}=1.417511$, the system experiences a saddle-node bifurcation around $E_5^{[SN]}=(0.216553, 1.227718, 0.347142)$. A slight further increase results in the occurrence of a Hopf bifurcation at $k_1^{[H]}=1.656706 $ around $E_5^{[H]}=( 1.736233, 0.545871, 2.371381 )$. The computed first Lyapunov coefficient is $\chi_{k_1}=2.4282446$ and hence the bifurcation is a subcritical Hopf. A final increase in fear levels results again in a saddle-node bifurcation at $k_1^{[SN]}=1.718081$ around $E_5^{[SN]}=( 1.272222, 0.697658, 1.919051)$. We observe that the densities of each population reduces as $k_1$ increases. For intermediate fear levels $k_1$, the occurrence of a Hopf bifurcation shows that when fear crosses a critical threshold, the populations experience sustained population oscillations. Ecosystem managers monitoring species in disease and fear-affected environments should consider the onset of cyclic population patterns as a signal to explore intervention strategies to prevent the populations from being ecologically vulnerable to extinction via a saddle-node bifurcation.


\begin{figure}[tbh]
\centering

\begin{minipage}{0.48\textwidth}
\centering
\includegraphics[width=\textwidth]{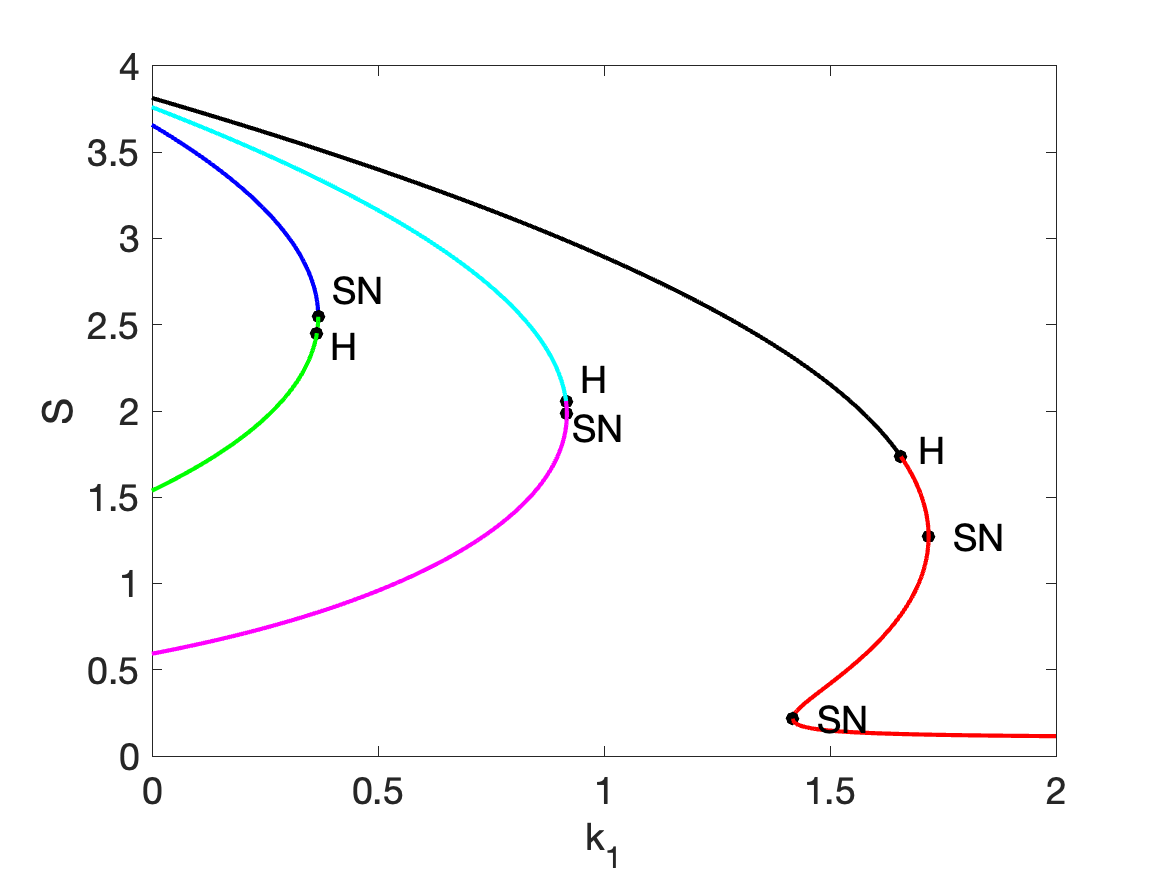}
\\[2mm]
$(a)$
\end{minipage}
\hfill
\begin{minipage}{0.48\textwidth}
\centering
\includegraphics[width=\textwidth]{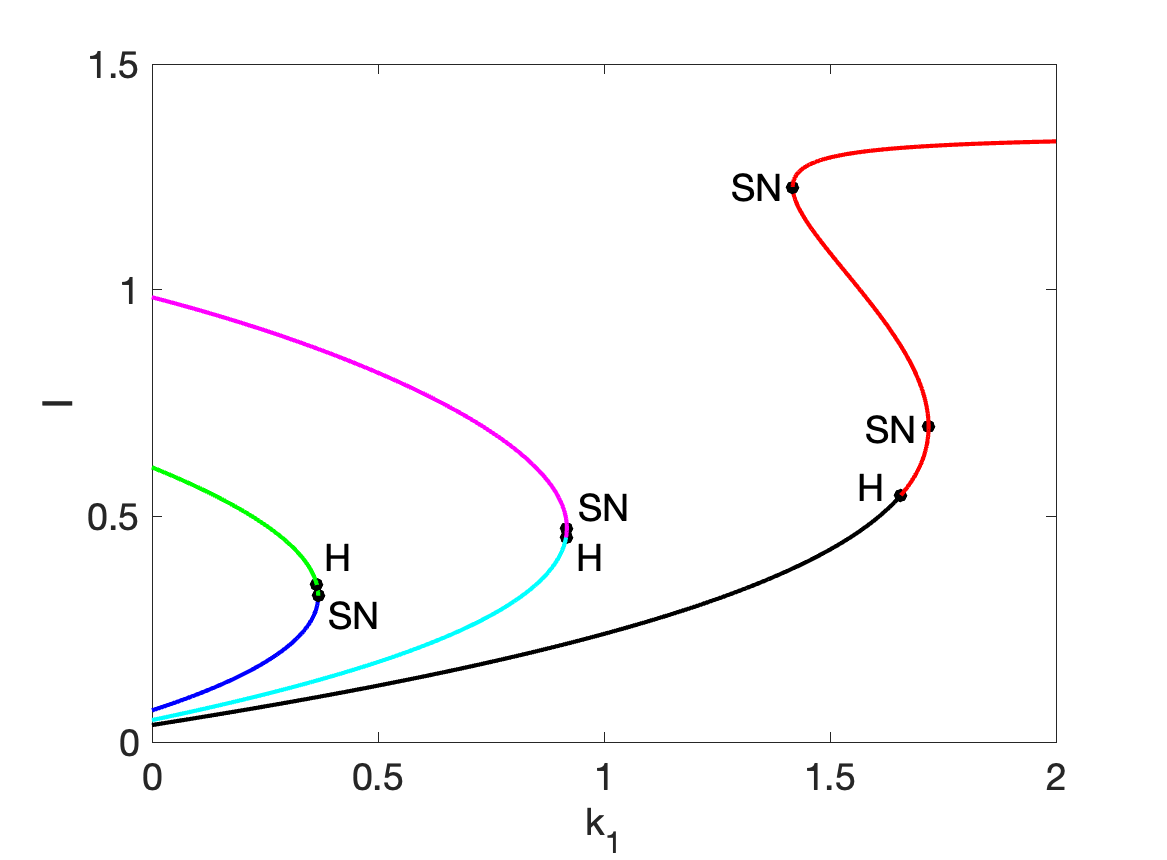}
\\[2mm]
$(b)$
\end{minipage}

\vspace{0.2cm}

\begin{minipage}{0.48\textwidth}
\centering
\includegraphics[width=\textwidth]{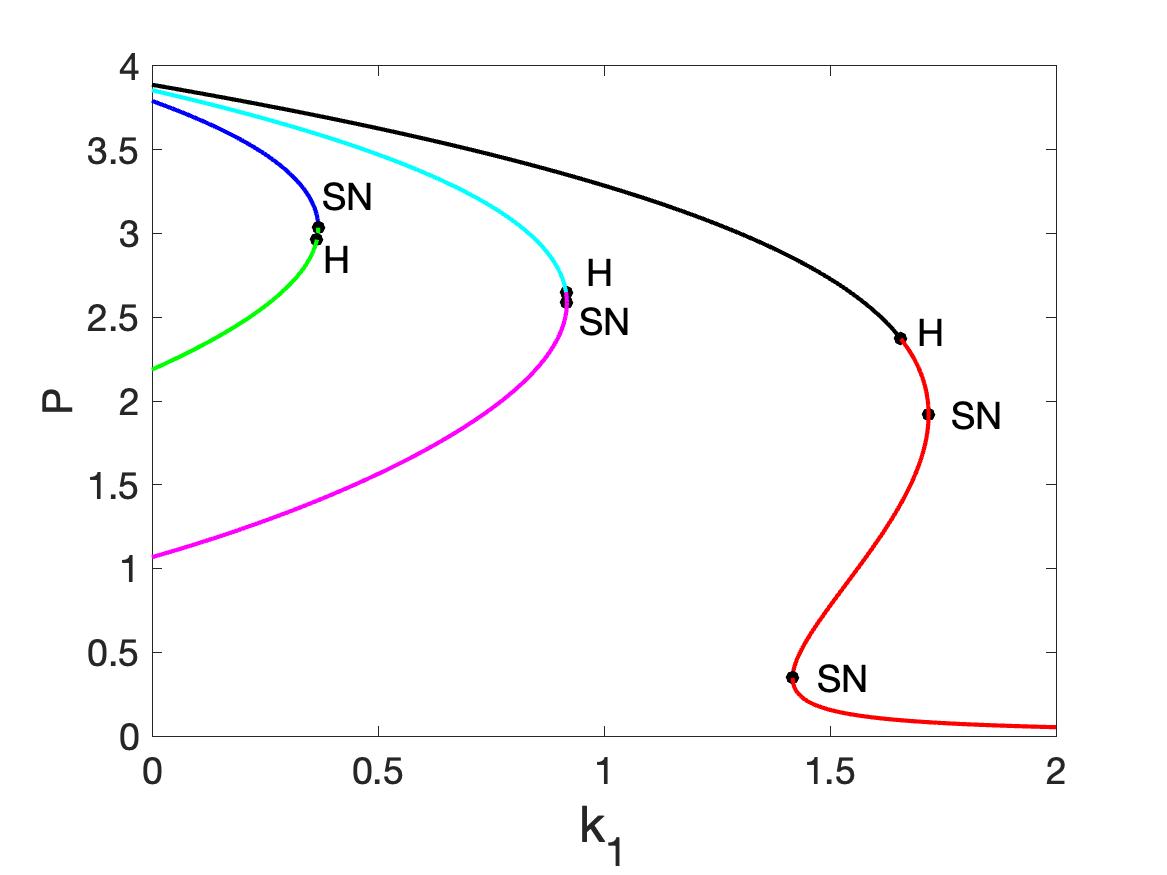}
\\[2mm]
$(c)$
\end{minipage}

\caption{Bifurcation diagram for model (\ref{Mainsystem}) showing the impact of predator-induced fear $k_1$ on susceptible prey across different Allee thresholds $L$. Fixed parameter values are set at $b_0=8$, $k_2=1$, and $d_3=0.5$. Stable branches are indicated by blue ($L=1$), cyan ($L=0$) and black ($L=-1$) lines, whereas unstable branches are represented by green ($L=1$), magenta ($L=0$) and red ($L=-1$) lines. Note: SN = Saddle-Node point, H = Hopf point.}
\label{Hbif2}
\end{figure}




\subsection{Two parameter bifurcation}
We explore a two-paramter bifurcation in the $(k_1,L)$ and $(k_1,k_2)$ planes respectively to deepen our understanding of how the interplay of fear-driven behavioral responses and Allee effects shape the long-term dynamics of the SIP system. We vary these pairs to reveal regions within the parameter space that support stable coexistence, oscillatory coexistence, predator-free state, infected prey-free state and finite-time extinction dynamics.

\subsubsection{Combined effect of $k_1$ and $L$}
The two parameter bifurcation plot in the $(k_1,L)$ plane seen in Figure \ref{bif_two_p1} reveals interesting ecological dynamics. When the Allee threshold is sufficiently negative, the susceptible prey population maintains a large buffer of its density where the system settles into a stable coexistence (SC) state among all three populations. As the fear parameter $k_1$ increases, the growth rate of susceptible prey is suppressed enough to cause a destabilization of the SC state. The SC region is then displaced by the finite-time extinction (FTE), predator-free equilibrium (PFE) and the oscillatory coexistence (OC) states. When the system enters the FTE regime, it collapses the susceptible prey population to zero in finite-time with a subsequent cascading effect on the extinction of infected prey and predator populations respectively. We also observed that, as $L$ increases towards zero, moderate fear levels suppressing reproduction is sufficient to trigger the collapse of the system. The emergence of the PFE state indicates the critical role played by the prey aggregation constant $r$ for intermediate fear levels $k_1$. By increasing $r$, the attack on susceptible prey is efficient and consequently reduces its population density. This reduction suppresses disease transmission between susceptible and infected individuals thereby reducing the population density of the infected class. The predator gains energy from both prey classes and so if the combined energy intake from both sources is very low, the predator population can no longer sustain itself and will go extinct. The appearance of the infected prey-free (IPF) state at $r=0.55$ for low fear levels $k_1$ and a strongly negative $L$ indicates that predators become more effective in attacking susceptible prey thereby reducing their population density. This in turn leads to a reduction at the rate at which new infections are produced and will eventually lead to the extinction of the infected class. The IPF region disappears at $r=0.5$ and $r=0.3$ because predatory attacks are weaker and the susceptible prey population remain at higher densities and will lead to increased contact with the infected class. In all, the $(k_1,L)$ shows that fear and Allee effect combined can act as extinction and destabilizing drivers.

\subsubsection{Combined effect of $k_1$ and $k_2$ }
The $(k_1,k_2)$ plane in Figure \ref{bif_two_p2} highlights the interplay between two distinct fear-driven suppression pathways. For a strongly negative Allee threshold $L=-3$, low fear levels for $k_1$ and $k_2$ enables all populations to persist. Increasing $k_1$ to suppress reproduction eventually drives the system to an FTE state as the susceptible prey population collapses. For high fear levels $k_1$ and moderate $k_2$ fear levels, the system transitions to an oscillatory coexistence (OC) regime as the stable coexistence (SC) state is destabilized via a Hopf bifurcation. At higher $k_2$, disease transmission is reduced. This reduces predator access to energetically profitable infected prey. With moderate prey aggregation and moderate predation, all populations can coexist. As $L$ increases towards $-1$, we observe that the SC region contracts and the FTE region expands. This demonstrates the key role of Allee effects on governing the ecological resilience and persistence of interacting populations. At $L=1$, the Allee threshold is positive and the entire $(k_1,k_2)$ plane is characterized by species extinction. This shows that, when Allee effect is strong, the collapse of the susceptible prey population is inevitable regardless of the two distinct fear levels for a chosen parameter regime.

\begin{figure}[tbh]
\centering

\begin{minipage}{0.48\textwidth}
\centering
\includegraphics[width=\textwidth]{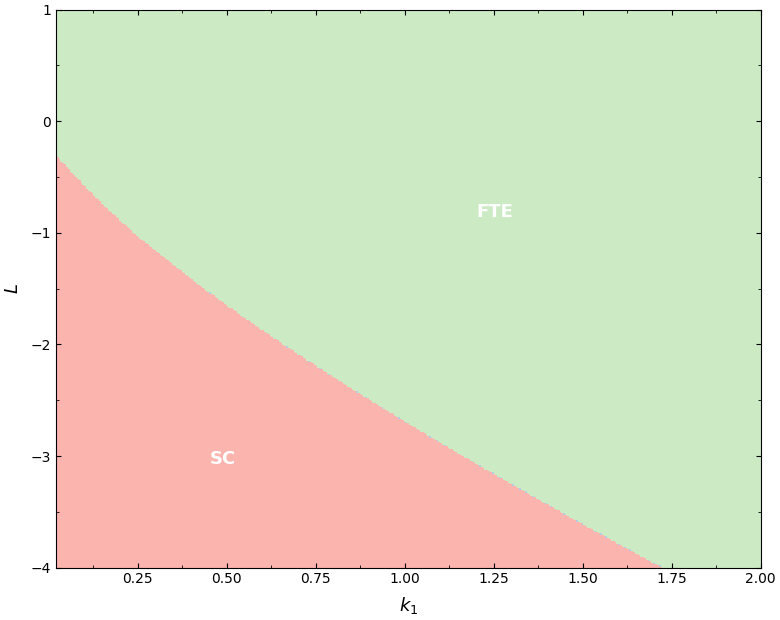}
\\[2mm]
$(a)$ $r=0.3$
\end{minipage}
\hfill
\begin{minipage}{0.48\textwidth}
\centering
\includegraphics[width=\textwidth]{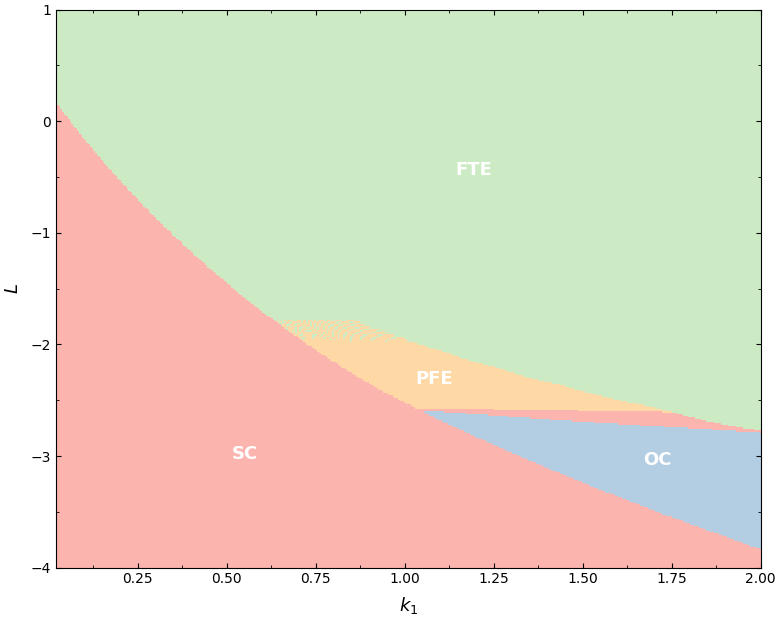}
\\[2mm]
$(b)$ $r=0.5$
\end{minipage}

\vspace{0.2cm}

\begin{minipage}{0.48\textwidth}
\centering
\includegraphics[width=\textwidth]{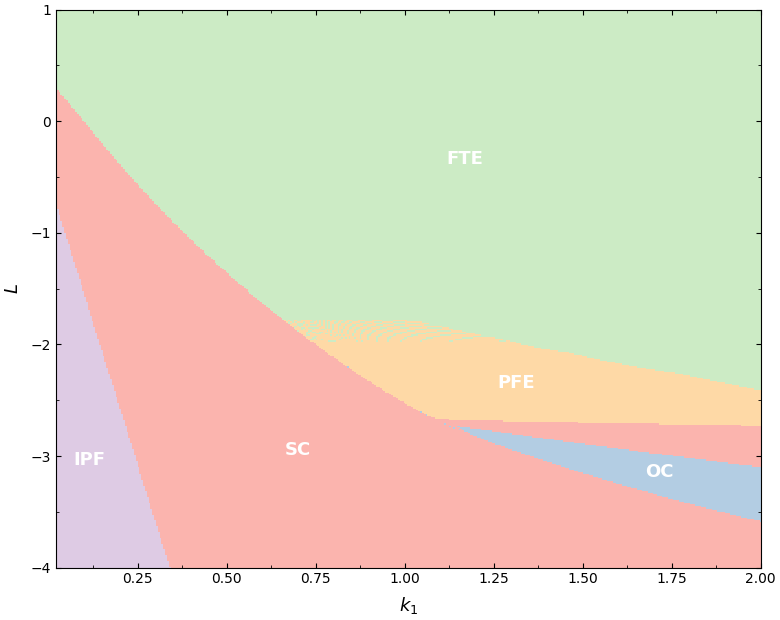}
\\[2mm]
$(c)$ $r=0.55$
\end{minipage}

\caption{Plots showing dynamical color coded regions of a two-parameter bifurcation in the $(k_1,L)$ space representing stable coexistence (SC), oscillatory coexistence (OC), finite-time extinction (FTE), predator-free (PFE) and infectious prey-free (IPF) states. Here we choose $b_0=2$ and $e_0=2$. All other parameters chosen are seen in Table \ref{tab:description}.}
\label{bif_two_p1}
\end{figure}

\begin{figure}[tbh]
\centering

\begin{minipage}{0.48\textwidth}
\centering
\includegraphics[width=\textwidth]{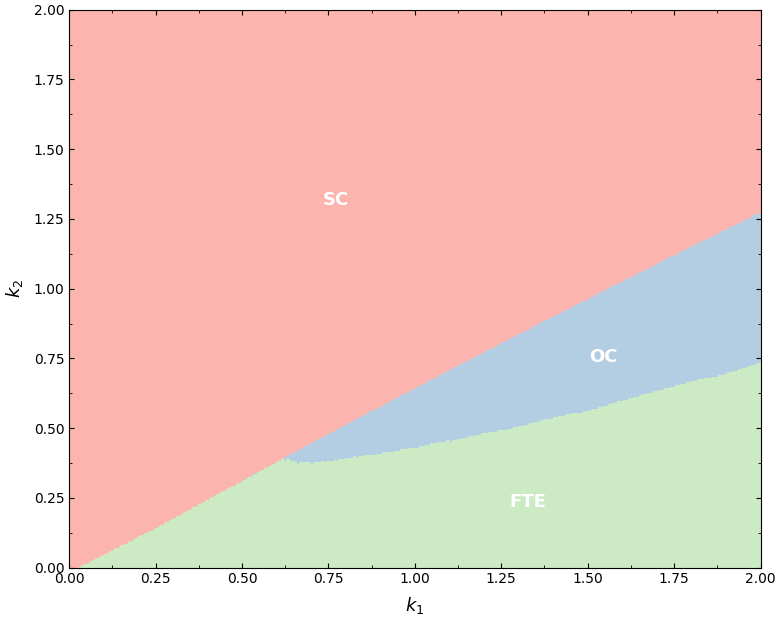}
\\[2mm]
$(a)$ $L=-3$
\end{minipage}
\hfill
\begin{minipage}{0.48\textwidth}
\centering
\includegraphics[width=\textwidth]{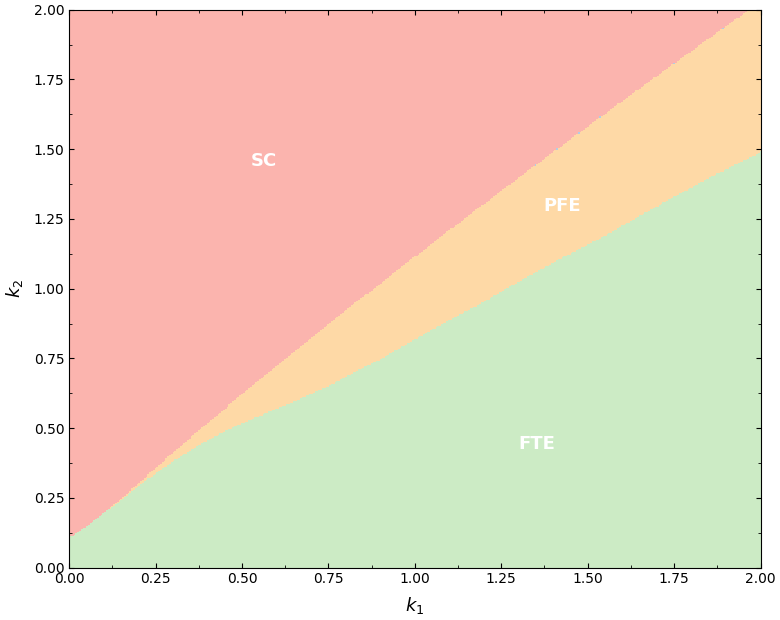}
\\[2mm]
$(b)$ $L=-2$
\end{minipage}

\vspace{0.2cm}

\begin{minipage}{0.48\textwidth}
\centering
\includegraphics[width=\textwidth]{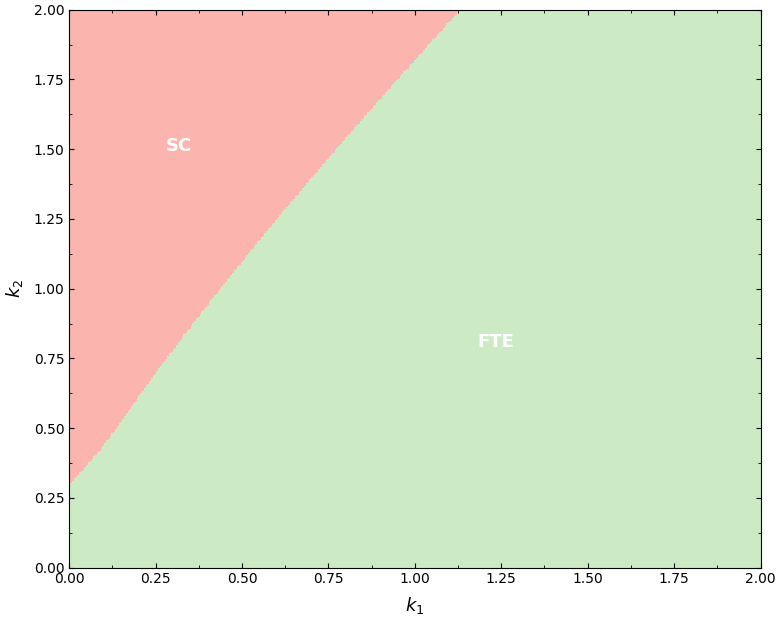}
\\[2mm]
$(c)$ $L=-1$
\end{minipage}
\hfill
\begin{minipage}{0.48\textwidth}
\centering
\includegraphics[width=\textwidth]{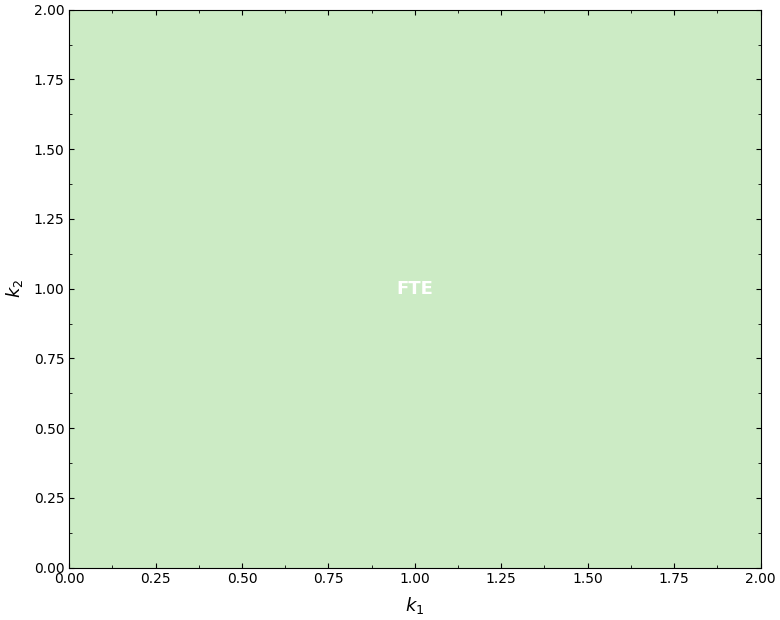}
\\[2mm]
$(d)$ $L=1$
\end{minipage}

\caption{Plots showing dynamical color coded regions of a two-parameter bifurcation in the $(k_1,k_2)$ space representing stable coexistence (SC), oscillatory coexistence (OC), finite-time extinction (FTE) and predator-free (PFE)  states. Here we choose $b_0=2$ and $e_0=2$. All other parameters chosen are seen in Table \ref{tab:description}.}
\label{bif_two_p2}
\end{figure}

\section{Bistability and Ecological Tipping Points}\label{sec:bistability and tipping}

\noindent The saddle-node bifurcation structure identified in Figure \ref{Bif1} suggests the existence of parameter regions in which multiple stable equilibria coexist. Such regimes are of particular ecological interest because they imply that long-term population outcomes may depend not only on environmental conditions but also on the initial state of the system.

Ecological systems often exhibit abrupt transitions between qualitatively distinct dynamical regimes when critical thresholds are crossed. Such transitions are commonly referred to as ecological tipping points. In general, a tipping point occurs when a small change in environmental conditions, demographic structure, or population densities produces a disproportionately large change in the long-term behavior of the ecosystem.

\begin{theorem}
\label{thm:SN_bistability}
Consider the model \eqref{Mainsystem} with $L$ chosen as a bifurcation parameter. Suppose that the following conditions hold:
\begin{enumerate}
    \item the model undergoes a generic saddle-node bifurcation of coexistence equilibria at $L=L^{[SN]}$;
    \item one branch of coexistence equilibria $E_5(L)$ created at the saddle-node is locally asymptotically stable for \(L\in(L^{[SN]},L^{[SN]}+\delta)\), for some \(\delta>0\);
    \item the predator-free equilibrium \(E_4(L)\) exists and remains locally asymptotically stable for \(L\in(L^{[SN]},L^{[SN]}+\delta)\).
\end{enumerate}
Then model \eqref{Mainsystem} is bistable for
\[
L\in(L^{[SN]},L^{[SN]}+\delta).
\]
In particular, there exist two disjoint open sets of initial conditions,
\(\mathcal{B}_4(L)\) and \(\mathcal{B}_5(L)\), such that solutions initiating in
\(\mathcal{B}_4(L)\) converge to \(E_4(L)\), while solutions initiating in
\(\mathcal{B}_5(L)\) converge to \(E_5(L)\). Consequently, the asymptotic state of the model depends on the initial population densities and may also depend on the history of parameter variation.
\end{theorem}

\begin{proof}
Since the model undergoes a generic saddle-node bifurcation of coexistence equilibria at \(L=L^{[SN]}\), the center-manifold reduction near the bifurcation point is locally topologically equivalent to the scalar normal form
\[
\dot{x}=\mu+\alpha x^2+\mathcal{O}(|x|^3+|\mu||x|),
\]
where \(\mu=L-L^{[SN]}\) and \(\alpha\neq 0\). Hence, on one side of the bifurcation value, there exist two nearby equilibrium branches. By assumption, one of these branches, denoted \(E_5(L)\), is locally asymptotically stable for \(L\in(L^{[SN]},L^{[SN]}+\delta)\).

By the third hypothesis, the predator-free equilibrium \(E_4(L)\) also exists and is locally asymptotically stable on the same interval. Therefore, for every fixed \(L\in(L^{[SN]},L^{[SN]}+\delta)\), the system possesses at least two distinct locally asymptotically stable equilibria, namely \(E_4(L)\) and \(E_5(L)\).

Local asymptotic stability implies the existence of open neighborhoods \(U_4(L)\) and \(U_5(L)\) of \(E_4(L)\) and \(E_5(L)\), respectively, such that solutions starting in \(U_4(L)\) converge to \(E_4(L)\), while solutions starting in \(U_5(L)\) converge to \(E_5(L)\). Since \(E_4(L)\neq E_5(L)\), these attracting neighborhoods may be chosen disjoint. Their forward invariant subsets are contained in distinct basins of attraction, denoted by \(\mathcal{B}_4(L)\) and \(\mathcal{B}_5(L)\).

Thus, for the same parameter values, different initial population densities can lead to different long-term ecological outcomes.

Finally, because one of the stable states is created or destroyed at the saddle-node threshold \(L=L^{[SN]}\), slow variation of \(L\) may cause the system to remain on one attracting branch until that branch loses existence at the saddle-node point. Reversing the parameter variation need not immediately return the system to its previous state, since the trajectory may then lie in a different basin of attraction. This path dependence is the dynamical signature of hysteresis. Hence, the asymptotic state depends not only on the current parameter value but also on the history of parameter variation.
\end{proof}

To investigate this phenomenon numerically, we consider parameter values within the bistable interval
$$L\in(-3.95198,-2.5164),$$
for which both the predator-free equilibrium $(E_4)$ and the coexistence equilibrium $(E_5)$ are locally stable. Numerical simulations were performed using two distinct sets of initial conditions while keeping all model parameters fixed.

Figure \ref{TSE4E5} illustrates the resulting trajectories. For the initial condition (3, 0.50, 1.50), the solution converges to the coexistence equilibrium $(E_5)$, indicating long-term persistence of susceptible prey, infected prey, and predators. In contrast, a second initial condition (0.25, 1.60, 0.05) leads to convergence toward the predator-free equilibrium $(E_4)$, resulting in predator exclusion despite identical environmental conditions and parameter values.

These results demonstrate the existence of \emph{alternative stable states} in the system. Consequently, small perturbations in population densities may alter the long-term ecological outcome by shifting the system from one basin of attraction to another. Such behavior is characteristic of ecological tipping points and highlights the sensitivity of predator--prey--disease systems to transient disturbances.

From a biological perspective, this finding suggests that predator-induced fear and aggregation may create conditions under which ecosystems exhibit reduced resilience. In particular, populations operating near the saddle-node threshold may undergo abrupt transitions between predator persistence and predator extinction following relatively small changes in population abundance. Therefore, the interaction between fear, aggregation, and Allee dynamics not only influences equilibrium densities but also governs the capacity of the ecosystem to recover from perturbations.

\begin{figure}[tbh]
\centering
\includegraphics[scale=0.48]{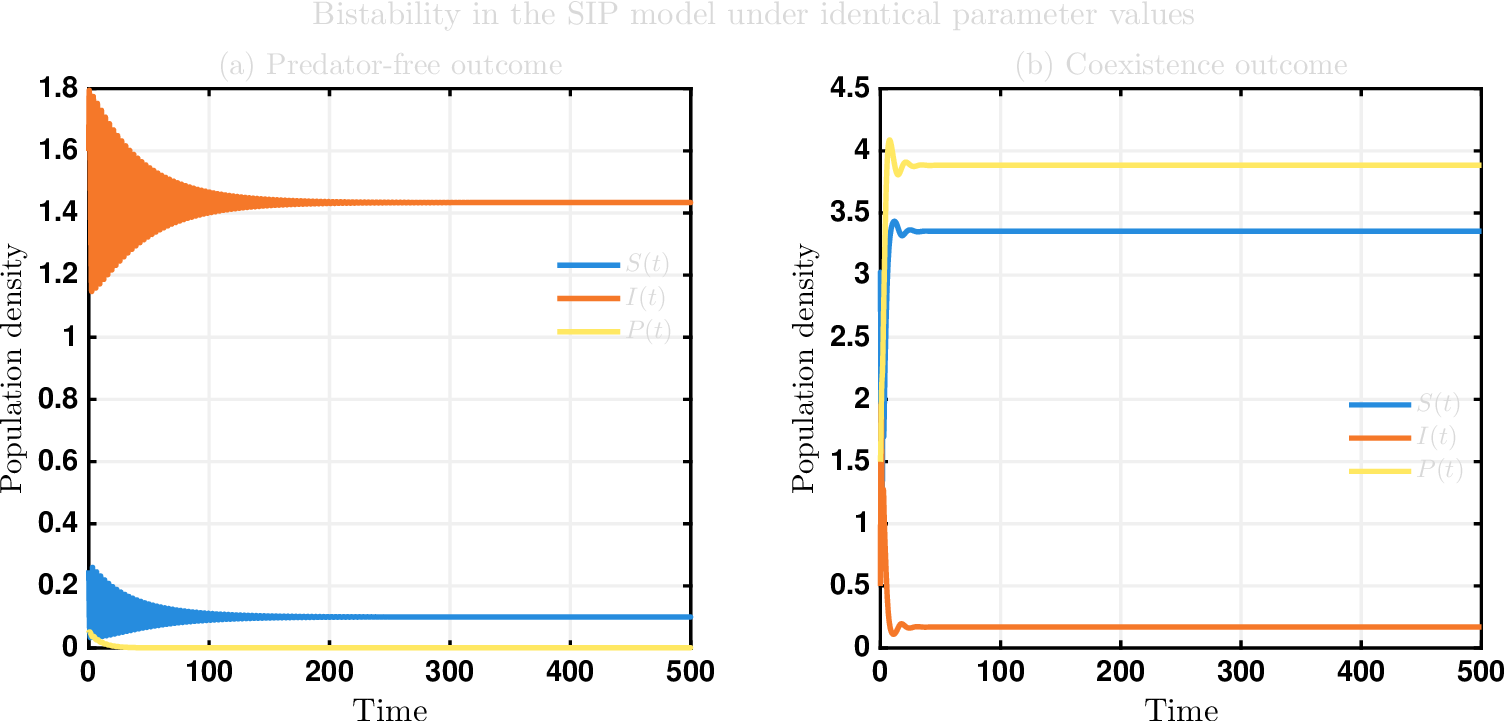}\\
\caption{Time series illustrating bistability in system~(1). The same parameter values are used in both panels with $L=-3$, which lies in the bistable interval identified in Figure \ref{Bif1}. Different initial conditions lead to distinct long-term outcomes: convergence to the predator-free equilibrium $(E_4) $ in panel (a), and convergence to the coexistence equilibrium $(E_5) $ in panel (b). This demonstrates the presence of alternative stable states and highlights the sensitivity of the system to initial population densities.}
\label{TSE4E5}
\end{figure}

\section{Finite-Time Extinction and Ecosystem Collapse}\label{sec:FTE}
\noindent The tipping phenomena described in Section \ref{sec:bistability and tipping} arise through changes in the underlying attractor structure of the system. We now investigate a different mechanism of ecosystem collapse. Specifically, we show that the aggregation-mediated predation term can force the susceptible prey population to reach the extinction state in finite time. This extinction event represents a state-induced ecological tipping point and subsequently initiates a cascade collapse of the infected prey and predator populations. The key mechanism responsible for this phenomenon is the non-Lipschitz predation term $d_0 S^r P$ with $0<r<1$, whose singular behavior near $S=0$ can force trajectories to hit the extinction state in finite time. This mechanism is further reinforced by the Allee effect when the susceptible prey density falls below the threshold $L$.

\begin{definition}[Finite-Time Extinction]
The susceptible prey population undergoes finite-time extinction if there exists a finite time $t_{ext}>0$ such that
\[
S(t)>0 \quad \text{for } 0\le t<t_{ext}, \qquad S(t_{ext})=0.
\]
\end{definition}

\noindent The following result is conditional in nature. It characterizes
the extinction dynamics once the susceptible prey population
enters the strong Allee region and predator pressure remains
uniformly positive.

\begin{theorem}
\label{thm:FTE_S}
Consider the model \eqref{Mainsystem} with $0<r<1$. Assume there exists a time interval $[0,t_{ext})$ such that
\[
0<S(t)<L, \qquad S(t)+I(t)\le K, \qquad P(t)\ge \underline{P}>0,
\quad \text{for all } t\in[0,t_{ext}),
\]
where $L>0$ and $\underline{P}$ is a positive constant. Then the susceptible prey
population $S(t)$ becomes extinct in finite time. More precisely,
\[
t_{ext} \le \frac{S(0)^{\,1-r}}{(1-r)d_0\underline{P}}.
\]
\end{theorem}

\begin{proof}
The susceptible prey equation in model \eqref{Mainsystem} is
\[
\frac{dS}{dt}
=
\frac{b_0S}{1+k_1P}\left(1-\frac{S+I}{K}\right)(S-L)
-d_0S^rP-\frac{e_0SI}{1+k_2P}.
\]
Since $0<S(t)<L$ and $S(t)+I(t)\le K$ on $[0,t_{ext})$, we have
\[
\left(1-\frac{S+I}{K}\right)\ge 0
\quad \text{and} \quad
(S-L)<0.
\]
Hence,
\[
\frac{b_0S}{1+k_1P}\left(1-\frac{S+I}{K}\right)(S-L)\le 0.
\]
Also,
\[
-\frac{e_0SI}{1+k_2P}\le 0.
\]
Therefore,
\[
\frac{dS}{dt}\le -d_0S^rP \le -d_0\underline{P}\,S^r.
\]
Set $c=d_0\underline{P}>0$. Then
\[
\frac{dS}{dt}\le -cS^r.
\]
For $0<r<1$, divide both sides by $S^r$ to obtain
\[
S^{-r}\frac{dS}{dt}\le -c.
\]
Since
\[
\frac{d}{dt}\left(\frac{S^{1-r}}{1-r}\right)=S^{-r}\frac{dS}{dt},
\]
we get
\[
\frac{d}{dt}\left(\frac{S^{1-r}}{1-r}\right)\le -c.
\]
Integrating from $0$ to $t$ yields
\[
\frac{S(t)^{1-r}-S(0)^{1-r}}{1-r}\le -ct.
\]
Hence
\[
S(t)^{1-r}\le S(0)^{1-r}-(1-r)ct.
\]
The right-hand side becomes zero at
\[
t=\frac{S(0)^{1-r}}{(1-r)c}
=
\frac{S(0)^{1-r}}{(1-r)d_0\underline{P}}=t_{bound}.
\]
Thus $S(t)$ reaches zero no later than this time.
\end{proof}

\noindent The extinction estimate derived above allows us to quantify how predator pressure and aggregation strength influence the rate at which the susceptible prey population approaches extinction.

\begin{corollary}
\label{cor:predator_pressure}
Under the assumptions of Theorem~\ref{thm:FTE_S}, the extinction time upper bound
\[
t_{ext} \le \frac{S(0)^{1-r}}{(1-r)d_0\underline{P}}
\]
decreases as either the attack rate $d_0$ increases or the lower bound on predator density $\underline{P}$ increases.
\end{corollary}

\begin{proof}
The conclusion follows immediately from the explicit upper bound for $t_{ext}$.
\end{proof}

\begin{corollary}
\label{cor:r_effect}
Under the assumptions of Theorem~\ref{thm:FTE_S}, smaller values of the aggregation exponent $r\in(0,1)$ lead to faster extinction of susceptible prey. In particular, the singular predation term $S^r$ is the mechanism responsible for the finite-time hitting of the extinction state.
\end{corollary}

\begin{proof}
The upper bound
\[
t_{ext}(r)\le \frac{S(0)^{1-r}}{(1-r)d_0\underline{P}}
\]
depends explicitly on $r$. Since $0<S(0)<L$ and $0<r<1$, stronger singularity near $S=0$ enhances the depletion rate of $S(t)$ and accelerates extinction.
\end{proof}

\noindent The next result shows that extinction of the susceptible prey population triggers a cascade collapse in the remaining populations.

\begin{corollary}
\label{cor:cascade}
Suppose the assumptions of Theorem~\ref{thm:FTE_S} hold, and let $t_{ext}$ be the finite time at which $S(t_{ext})=0$. Then the infected prey population and predator population decay asymptotically to zero; that is,
\[
\lim_{t\to\infty} I(t)=0,
\qquad
\lim_{t\to\infty} P(t)=0.
\]
\end{corollary}

\begin{proof}
For $t\ge t_{ext}$, we have $S(t)=0$. The infected prey equation becomes
\[
\frac{dI}{dt}=-d_1IP-a_1I \le -a_1I.
\]
By comparison,
\[
I(t)\le I(t_{ext})e^{-a_1(t-t_{ext})},
\]
and hence $\lim_{t\to\infty}I(t)=0$.

Next, for $t\ge t_{ext}$, the predator equation reduces to
\[
\frac{dP}{dt}=d_3IP-a_2P = P(d_3I-a_2).
\]
Since $I(t)\to 0$, there exists $T\ge t^*$ such that $d_3I(t)<a_2/2$ for all
$t\ge T$. Therefore,
\[
\frac{dP}{dt}\le -\frac{a_2}{2}P,
\qquad t\ge T.
\]
By comparison,
\[
P(t)\le P(T)e^{-\frac{a_2}{2}(t-T)},
\]
so $\lim_{t\to\infty}P(t)=0$.
\end{proof}

\begin{figure}[tbh]
\centering
\includegraphics[scale=0.5]{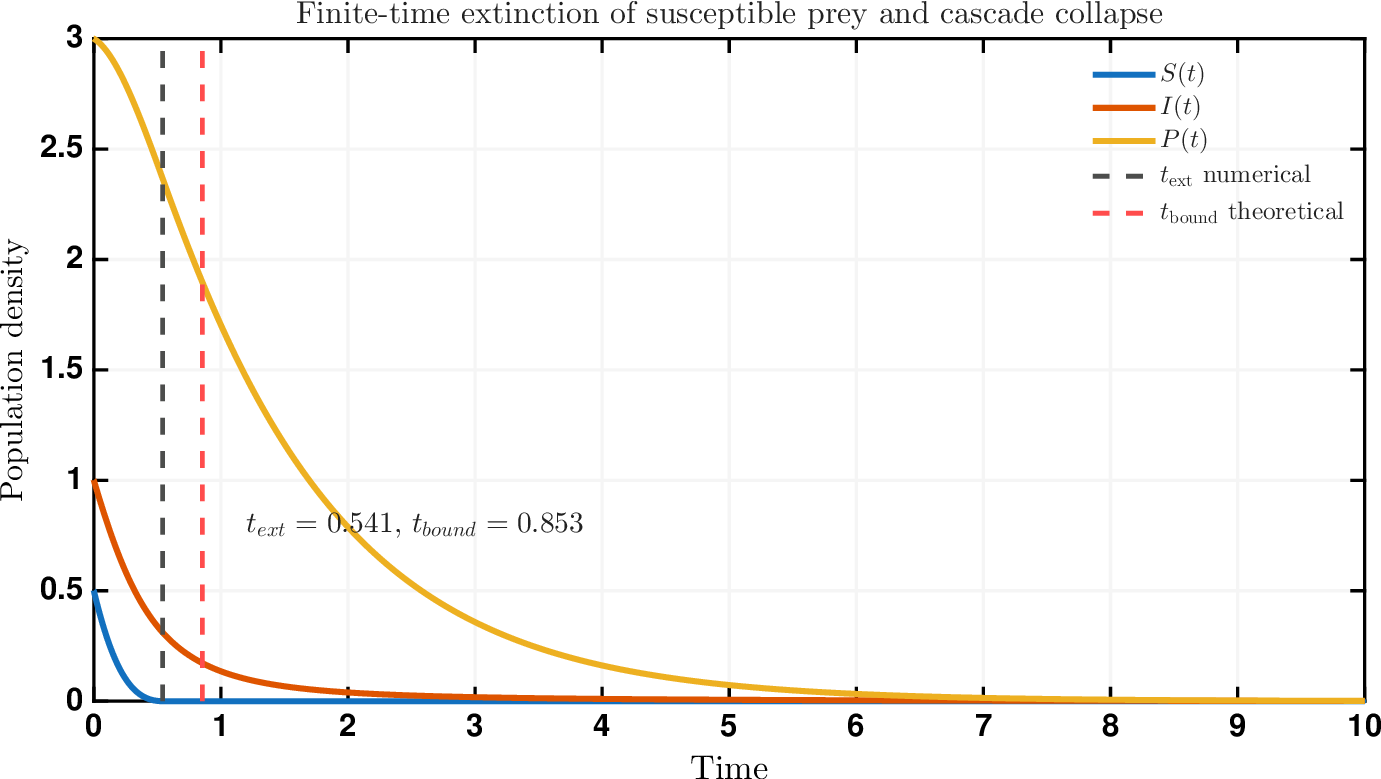}\\
\caption{Numerical illustration of finite-time extinction and cascade collapse in model \eqref{Mainsystem}. The initial condition satisfies \(0<S(0)<L\), with \((S(0),I(0),P(0))=(0.5,1,3)\) and \(L=1\). The susceptible prey population \(S(t)\) reaches the extinction threshold in finite time. The vertical black dashed line denotes $t_{ext}$, while the red dashed line denotes the theoretical upper bound $t_{bound}$. Following the extinction of $S(t)$, both $I(t)$ and $P(t)$ decay toward zero, illustrating cascade collapse.}
\label{fig:TS_FTE}
\end{figure}

\noindent To illustrate Theorem \ref{thm:FTE_S} and Corollary \ref{cor:cascade}, we numerically simulate model \eqref{Mainsystem} under a strong Allee threshold using an initial condition satisfying
\[
0<S(0)<L, \qquad S(0)+I(0)\le K.
\]
Specifically, we use
\[
(S(0),I(0),P(0))=(0.5,1,3),
\]
with \(L=1\) and \(K=4\). The numerical extinction time \(t_{\mathrm{ext}}\) is computed as
the first time at which \(S(t)\) reaches a small extinction tolerance by using MATLAB R2026a. To compare
with the theoretical estimate in Theorem \ref{thm:FTE_S}, we compute
\[
\underline{P}=\min_{0\le t\le t_{\mathrm{ext}}}P(t),
\]
and use the upper bound, $t_{bound}$. Figure \ref{fig:TS_FTE} shows that the susceptible prey population reaches the extinction state in finite time, and the computed value of \(t_{\mathrm{ext}}\) lies below the theoretical upper bound. After the susceptible prey population collapses, the infected prey and predator populations subsequently decay to zero, confirming the cascade extinction mechanism described in Corollary \ref{cor:cascade}.

\begin{remark}
The finite-time extinction established in Theorem \ref{thm:FTE_S} does not arise from the Allee effect or aggregation mechanism alone. Rather, it emerges from the interaction between demographic threshold effects and aggregation-based predation. The Allee effect creates conditions under which population recovery becomes difficult, while the non-Lipschitz predation term accelerates depletion near the extinction state. Together, these mechanisms drive the susceptible prey population to extinction in finite time. When the susceptible prey density falls below the Allee threshold $L$, reproduction becomes insufficient to offset losses, causing the population to decline. At the same time, the non-Lipschitz predation term $d_0S^r\underline{P}$ with $0<r<1$ accelerates the depletion of the susceptible prey population near $S=0$, forcing the trajectory to reach the extinction state in finite time. The explicit upper bound  
\[
t_{ext}\le t_{bound}
\]
further reveals that stronger predator pressure (larger $d_0$ or $P$) reduces the extinction time. Thus aggregation-mediated predation may paradoxically accelerate ecosystem collapse once fear-driven behavioral responses and demographic threshold effects have pushed the susceptible prey population below the Allee threshold.
\end{remark}

\section{Conclusion}\label{sec:Conclusion}
\noindent This study has established that predator-induced fear, prey aggregation, and Allee dynamics are not merely additive modifiers of predator-prey-disease interactions but interacting mechanisms capable of generating qualitatively distinct and more dangerous dynamical behavior than any pairwise combination produces. The central finding — that fear-mediated aggregation can paradoxically accelerate ecosystem collapse by amplifying demographic vulnerability at low population densities — has been established both analytically and numerically, and connects the paper's mathematical results to the broader literature on ecological tipping points and critical transitions \cite{Courchamp2008,Scheffer2001,Scheffer2009,Dakos2019}.

The rigorous derivation of finite-time extinction (Theorem \ref{thm:FTE_S}) establishes that the non-Lipschitz aggregation term $d_0 S^r P$, $0 < r < 1$, is not merely a mathematical convenience but a mechanistically important feature of aggregation-based predation that generates qualitatively different extinction dynamics from smooth models. The explicit upper bound on extinction time quantifies for the first time how predator pressure and aggregation strength jointly determine the speed of ecological collapse, providing a potential tool for assessing extinction risk in threatened populations. The cascade collapse result (Corollary \ref{cor:cascade}) demonstrates that this single-population extinction event is sufficient to drive the entire ecological community to extinction, revealing the fragility of the food web structure when the susceptible prey population serves as the foundational energy source for both disease dynamics and predator persistence \cite{antwi2025dynamics}.

The bifurcation analysis reveals a rich structure of dynamical transitions including transcritical, saddle-node, and Hopf bifurcations, and the two-parameter continuation identifies distinct parameter regimes supporting stable coexistence, oscillatory coexistence, predator exclusion, and finite-time extinction \cite{Kuznetsov2004,Scheffer2001}. The monotone expansion of the finite-time extinction region with decreasing $r$ demonstrates that aggregation strength is a qualitative determinant of whether a given set of fear and Allee parameters leads to persistence or collapse — a finding with direct implications for the management of threatened species in which aggregation behavior is observable and potentially modifiable \cite{Courchamp2008,KrauseRuxton2002}.

From a conservation perspective, the results suggest that 
populations simultaneously subject to infectious disease, 
predation pressure, and low-density demographic vulnerability 
may be far more fragile than models treating these factors 
independently would predict \cite{Brook2008}. In particular, the bistability results demonstrate that populations operating near saddle-node thresholds may undergo irreversible transitions to predator exclusion or extinction following perturbations that would be considered minor under single-mechanism models \cite{Scheffer2001,Scheffer2009}. Early-warning 
indicators based on asymptotic dynamics — such as critical 
slowing down — may fail to detect the approach to finite-time 
collapse entirely, since the collapse mechanism operates through 
a non-Lipschitz singularity rather than through a loss of 
stability in the classical sense \cite{Boerlijst2013,Dakos2012}.

Several directions remain open for future investigation. The 
incorporation of spatial dispersal would allow examination of whether aggregation-driven collapse can be arrested by immigration 
from refuge populations, a question directly relevant to  metapopulation management and rescue effects \cite{Levins1969,Hanski1998}. Environmental stochasticity may interact with the Allee threshold to generate extinction probabilities that could be compared with empirical data from 
monitored wildlife populations \cite{Lande1993,Melbourne2008}. Adaptive behavioral responses — in which aggregation intensity itself varies with perceived predation risk — would introduce a feedback between fear and aggregation that could either stabilize or destabilize the system depending on the timescale of behavioral adaptation relative to population dynamics
\cite{Lima1990,Brown1999}. Finally, the development of early-warning indicators specifically designed to detect aggregation-driven finite-time collapse represents both a mathematical challenge and a practical conservation priority \cite{Scheffer2009,Dakos2012}.

\section*{Declaration of competing interest}
\noindent The authors declare that they have no known competing financial interests or personal relationships that could have appeared to influence the work reported in this
paper.

\section*{Data availability}
\noindent Not applicable.

\bibliographystyle{unsrt}
\bibliography{FearAllee}
\end{document}